\documentclass[submission, Phys]{SciPost}

\pdfoutput=1
\usepackage{euscript,amsmath,amsthm,amsfonts,amssymb}
\usepackage{mathtools}
\usepackage{float}
\usepackage{outlines}
\usepackage{adjustbox}
\usepackage{longtable}
\usepackage{orcidlink}
\usepackage{tabularx}  
\usepackage[table]{xcolor}
\usepackage{makecell}
\usepackage{tikz}
\usepackage{arydshln}
\usepackage{graphicx}
\usepackage{subfig}
\usepackage{subcaption}
\usepackage{graphicx}

\definecolor{MIblue}{RGB}{0,122,183}
\definecolor{MIred}{RGB}{239,32,12}
\definecolor{MIgreen}{RGB}{27,179,0}
\definecolor{MIpurple}{RGB}{132,46,196}

\definecolor{MIcyan}{RGB}{0,174,245}
\definecolor{MIlime}{RGB}{91,215,49}
\definecolor{MIyellow}{RGB}{255,211,24}

\newcommand{\ket}[1]{|{#1}\rangle}

\newcommand{\bra}[1]{\langle{#1}|}

\newcommand{\be}{\begin{equation}}
\newcommand{\ee}{\end{equation}}
\newcommand{\bea}{\begin{eqnarray}}
\newcommand{\eea}{\end{eqnarray}}

\newcommand{\ba}{\begin{eqnarray}}
\newcommand{\ea}{\end{eqnarray}}

\newcommand{\la}{\langle}
\newcommand{\lb}{\rangle}

\def\mod{\ {\rm mod}\ }

\newtheorem{theorem}{Theorem}
\newtheorem{definition}{Definition}

\newcommand{\tq}{{\mathtt q}}

\def\CE{{\mathcal E}}

\hypersetup{
    colorlinks,
    linkcolor={red!50!black},
    citecolor={blue!50!black},
    urlcolor={blue!80!black}
}

\newcommand{\ClusterColoringBase}{%
  \path[use as bounding box]
    (-1.95,-2.18) rectangle (4.72,1.30);

  \foreach \j/\xx/\cc in {
    1/0.67/MIcyan,
    2/1.78/MIred,
    3/2.89/MIlime,
    4/4.00/MIyellow%
  }{
    \node at (\xx,1.06) {$A_{\j}$};
    \draw[\cc,line width=0.8pt]
      (\xx,0.48) -- (\xx,0.88);
  }

  \draw[line width=0.6pt]
    (0.08,-2.15) -- (0.08,0.35) -- (4.62,0.35);

  \foreach \row/\bits/\statecol/\la/\lb/\ra/\rb in {
    0/0000/MIblue/MIblue/MIred/MIblue/MIgreen,
    1/0011/MIred/MIblue/MIred/MIred/MIpurple,
    2/1100/MIgreen/MIgreen/MIpurple/MIblue/MIgreen,
    3/1111/MIpurple/MIgreen/MIpurple/MIred/MIpurple%
  }{
    \pgfmathsetmacro{\yy}{-0.60*\row}

    \node[anchor=east,inner sep=0pt]
      at (-0.62,\yy) {$|\bits\rangle$};
    \fill[\statecol]
      (-0.26,\yy) circle[radius=0.12];

    \foreach \xx in {0.67,1.78}{
      \fill[\la]
        ({\xx-0.17},\yy) circle[radius=0.12];
      \fill[\lb]
        ({\xx+0.17},\yy) circle[radius=0.12];
    }

    \foreach \xx in {2.89,4.00}{
      \fill[\ra]
        ({\xx-0.17},\yy) circle[radius=0.12];
      \fill[\rb]
        ({\xx+0.17},\yy) circle[radius=0.12];
    }
  }
}

\begin{document}
\begin{center}{\Large \textbf{A group theoretic calculation \\of multi-invariants of some stabilizer states
}}\end{center}

\begin{center}
Bartłomiej Czech\,\orcidlink{0000-0002-6368-3502}\textsuperscript{1},
Yichen Feng\,\orcidlink{0009-0003-2816-2714}\textsuperscript{1},
Jonathan Harper\,\orcidlink{0000-0003-4913-6470}\textsuperscript{2,3},
Xin-Xiang Ju\,\orcidlink{0000-0002-9905-3908}\textsuperscript{1,4},
\end{center}

\begin{center}
{\bf 1} Chen-Ning Yang Institute for Advanced Study, Tsinghua University, Beijing, China
\\ 
{\bf 2} Department of Physics,
Case Western Reserve University,\\ 10900 Euclid Ave, Cleveland, OH 44106 USA
\\
{\bf 3} Center for Gravitational Physics and Quantum Information,\\ Yukawa Institute for Theoretical Physics, Kyoto University,\\ Kitashirakawa Oiwakecho, Sakyo-ku, Kyoto 606-8502, Japan
\\
{\bf 4} School of Physical Sciences, University of Chinese Academy of Sciences, Beijing, China
\\
Email: {\small \sf \href{mailto:bartlomiej.czech@gmail.com}{bartlomiej.czech@gmail.com}},   {\small \sf \href{mailto:fyc23@mails.tsinghua.edu.cn}{fyc23@mails.tsinghua.edu.cn}},   {\small \sf \href{mailto:jonathan.harper@case.edu}{jonathan.harper@case.edu}},  {\small \sf \href{mailto:juxinxiang21@mails.ucas.ac.cn}{juxinxiang21@mails.ucas.ac.cn}}
\end{center}

\section*{Abstract}
\noindent 
We study multi-invariants---a multi-partite generalization of R{\'e}nyi entropies---of simple stabilizer states, for an arbitrary replica symmetry group. We develop a group theoretic technique for evaluating such invariants, which operates on the replica group and involves counting intersections among its cosets. Our techniques streamline the calculations of multi-invariants for graph states with two GHZ-like nodes and for certain generalizations thereof. The methodology developed here is most suitable when the quantum state is simple but flexibility over the replica group is desired.

\vspace{10pt}
\noindent\rule{\textwidth}{1pt}
\tableofcontents\thispagestyle{fancy}
\noindent\rule{\textwidth}{1pt}
\vspace{10pt}

\section{Introduction}
Quantifying quantum entanglement is without doubt a foundationally important problem in theoretical physics~\cite{HorodeckiEtAl2009}. While purely bipartite entanglement is fully classified by the entanglement spectrum, there is not a one-size-fits-all measure of multi-partite entanglement. This paper studies one recently proposed such measure, which---in a sense that we make precise below---is the closest possible analogue of the bipartite entanglement spectrum.

For a pure state $\ket{\Psi}$, a large class of entanglement measures are linear functionals $s^\vee$ applied to a number of copies of $\ket{\Psi}$ and of $\bra{\Psi}$~\cite{LeiferLindenWinter2004,CarrozzaChevrierLionni2026}:
\begin{equation}
    \mathcal{S} = s^\vee \Big( \ket{\Psi}^{\otimes s_k} \bra{\Psi}^{\otimes s_b} \Big)
\label{generalmeasure}
\end{equation}
Applied to states of bipartite systems $A \overline{A}$, this class includes the moments of the spectrum of $\rho_A = {\rm Tr}_{\overline{A}}\, \ket{\Psi} \bra{\Psi}$
\begin{equation}
    e^{(1-\alpha) S_\alpha} \equiv {\rm Tr}_A
    \Big( {\rm Tr}_{\overline{A}} \ket{\Psi} \bra{\Psi} \Big)^\alpha
\label{renyis}
\end{equation}
and---since the spectral density can be extracted from its moments in standard ways---the spectrum itself. A further generalization of (\ref{generalmeasure}), which involves explicit extremization~\cite{Vidal2000monotones}, essentially covers all measures of entanglement.

Among all measures~(\ref{generalmeasure}), the moments~(\ref{renyis})---here written as $e^{(1-\alpha) S_\alpha}$ using the definition of the R{\'e}nyi entropy $S_\alpha$---stand out as the simplest. If we represent the bipartite state
\be
    \ket{\Psi}_{A\overline{A}} = \sum_{i=1}^{d_A} \sum_{j=1}^{d_{\overline{A}}} 
    \psi_{ij}\, \ket{e^{i}_A} \otimes \ket{e^{j}_{\overline{A}}}
\ee 
graphically as a single-node `tensor network' \textemdash$\bullet$\textemdash~and similarly use \textemdash$\circ$\textemdash~for $_{A\overline{A}}\bra{\Psi}$ then the moments are simply:
\begin{equation}
    e^{(1-\alpha) S_\alpha} =\,
    \underbrace{\textrm{\textemdash$\bullet$\textemdash$\circ$\textemdash}
    ~\ldots~
\textrm{\textemdash$\bullet$\textemdash$\circ$\textemdash}\phantom{\big|}}_{\text{$\alpha$ copies, cyclically connected}}
\label{eq:moments}
\end{equation}
In this expression, we treat $\ket{\Psi}$~and $\bra{\Psi}$ as jigsaw puzzle pieces with two connectors each (one for subsystem $A$ and one for $\overline{A}$), then form a compound with all connectors joint. Once more, we emphasize that in finite-dimensional Hilbert spaces the bipartite entanglement spectrum is a routine reconfiguring of these quantities. 

This suggests a straightforward generalization to multi-partite states: treat them as puzzle pieces with three or more connectors and build closed structures. For example, by using copies of a three-partite state $\ket{\Psi}_{ABO}$ (respectively $_{ABO}\bra{\Psi}$) as jigsaw pieces with three distinct connectors, we can assemble the structure in shown in Figure~\ref{fig:intro3partyjigsaw}.
\begin{figure}[t]
  \centering
   \includegraphics[width=.8\textwidth,page=42]{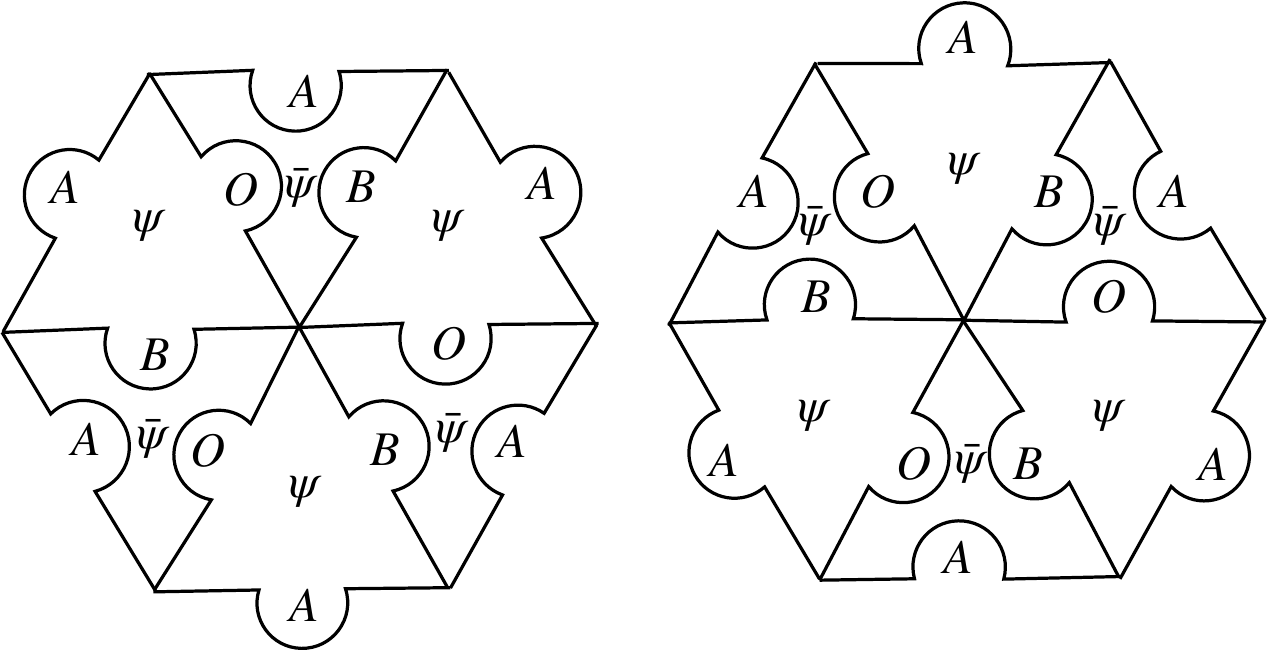}
  \caption{A jigsaw-like structure, which can be assembled from 6 copies of $\ket{\Psi}_{ABO}$ and 6 copies of $_{ABO}\bra{\Psi}$. The two panels are to be pasted together as top and bottom to form a topological sphere.}
\label{fig:intro3partyjigsaw}
\end{figure}
\noindent
Later in the text, we redraw the same quantity using a more versatile notation in Figure~\ref{fig:all1}. The main purpose of this paper is to develop an understanding of such quantities in a manageable yet interesting class of quantum states---that is, stabilizer states.

Somewhat surprisingly, the jigsaw-like quantities exemplified by Figure~\ref{fig:intro3partyjigsaw} were only introduced in the literature fairly recently---and in a rather different language \cite{GaddeKrishnaSharma2023,2025JHEP...06..116G}. The primary motivation there was an application to the AdS/CFT correspondence \cite{Maldacena1997} or holographic duality. Applied on the CFT side of the duality, a certain limit of these quantities---analogous to the von Neumann entropy limit of the R{\'e}nyi entropies~\cite{LewkowyczMaldacena2013}---was proposed to compute the minimal area of a branching surface~\cite{GaddeKrishnaSharma2022}, which divides the dual AdS spacetime into components. This proposal is a generalization of the famous Ryu-Takayanagi relation \cite{RyuTakayanagi2006} and its R{\'e}nyi predecessor \cite{Dong2016}.

Yet quantities like Figure~\ref{fig:intro3partyjigsaw} deserve scrutiny also---perhaps primarily---away from the holographic setting~\cite{BerthiereGaudin2026}. In going from equation~(\ref{eq:moments}) to Figure~\ref{fig:intro3partyjigsaw}, we replaced an essentially bipartite structure with a tripartite one. This lends credence to the expectation that such quantities can be probes of genuine multi-partite entanglement~\cite{IizukaNishida2025,IizukaLinNishida2025,Ju:2026xfh,Gadde2026Signals}. Moreover, if one has at one's disposal an efficient representation of the state in the form of a tensor network, such quantities are relatively easy to evaluate~\cite{PeningtonWalterWitteveen2023,HuLinNechita2026}. 

\paragraph{The purpose of this paper} is to examine a selection of quantities like Figure~\ref{fig:intro3partyjigsaw}, applied to the simplest stabilizer states and their close cousins. The quantities on which we focus have been called `multi-invariants' \cite{2025JHEP...06..116G}, a nomenclature that imposes a few further constraints on the jigsaw pattern, to be explained in Section~\ref{sec:multis}. We choose stabilizer states as a testing ground of multi-invariants~\cite{AkellaGaddePandey2026,CzechFengWuXie2026} because they are easy to study, widely known, have many applications in quantum computing and quantum error corrections, and their bipartite entanglement spectra are flat \cite{fattal2004stabilizer,hein2006entanglementgraphstatesapplications}. This last property implies that, should our calculations yield non-trivial answers, we will have pinned down some form of multi-partite entanglement.

Computing multi-invariants of stabilizer states knits together two structural elements: the defining jigsaw pattern and the stabilizer group of state $\ket{\Psi}$. In another paper, two of us (BC and YF) and collaborators presented related but more general calculations, which are centered around the stabilizer group of $\ket{\Psi}$ \cite{CzechFengWuXie2026}. When those calculations are applied to the task in this paper, the jigsaw pattern recedes to a secondary or even lesser importance. In contrast, the present paper pursues a different strategy, which turns out to be more beneficial when the underlying quantum state has a simple and rigid form. Our strategy centralizes the group theoretic properties of the jigsaw pattern, which is algebraically represented by {\bf a replica group}.

A computation that focuses on the pattern of contractions (and keeps it adjustable) is more germane to the motivation of this paper, which is to probe the utility of multi-invariants for characterizing multi-partite quantum entanglement. 
Unfortunately, our group-theoretic approach---while versatile with respect to the choice of jigsaw contractions---is only workable when the state $\ket{\Psi}$ is simple enough, as quantified by the number of non-vanishing components of its wavefunction in some natural basis. In other words, here we make a deliberate choice to accommodate broadly varying jigsaw patterns, at the cost of limiting the underlying states. In another paper \cite{SMI2}, which is being submitted to the arXiv concurrently, we make the opposite choice: accommodating a broader class of states under a rigid pattern of contractions, described by a replica symmetry $\mathbb{Z}_n^{q-1}$. The two approaches are complementary; we contrast them further in the Discussion.

To wit, the states considered in this paper are GHZ states~\cite{GHZ1989} on arbitrarily many parties as well as graph states with two GHZ-like centers. All these states have wavefunctions, which are simple in the requisite sense. They are also otherwise interesting because of the role they play in quantum computing \cite{RaussendorfBriegel2001,hein2006entanglementgraphstatesapplications} and CHSH inequalities \cite{Bell1964,CHSH1969,Mermin1990,hein2006entanglementgraphstatesapplications}. This makes them an attractive arena for testing Figure~\ref{fig:intro3partyjigsaw} and generalizations.

In variants of Figure~\ref{fig:intro3partyjigsaw} whose jigsaw pattern is haphazard and unstructured our calculational techniques do not lead to simplifications. When the contraction adheres to a certain organizing principle, however, the calculation reveals a deep and rewarding connection to the theory of finite groups. The group theoretic methodology harks back to the original, holographic motivation for considering these quantities: the pattern of contractions among indices of $\ket{\Psi}$ and $\bra{\Psi}$ should follow a permutation, which is the automorphism group of a tessellated two-dimensional manifold. Figure~\ref{fig:intro3partyjigsaw} manifestly satisfies this demand. While our calculations make use of this holographically motivated device, we have taken special care to write the paper in an entirely self-contained manner, accessible to readers with no familiarity with the AdS/CFT correspondence. Still, some constructions in this paper might hold extra appeal to readers familiar with holographic duality, in comparison to non-holographer colleagues.

\paragraph{Organization} 
Section~\ref{sec:multis} introduces multi-invariants and explains how they can be described geometrically as tessellations. In Section \ref{sec:warm-up}, we recast certain multi-invariant calculations on stabilizer states as a ``coloring problem" on the underlying tesselation. We provide two specific examples to build intuition: the GHZ state and the four-party cluster state. 
Section \ref{sec:coloringmethod} contains our main results: it explains how group theory considerations organize the coloring problem set up in Section~\ref{sec:warm-up}. In Section \ref{sec:moreparty}, we apply our techniques to multi-invariants of subsystems, which contain more than one qubit. Section \ref{sec:discussion} is the Discussion. Some technical facts from algebra are collected in the appendices.

\medskip

\paragraph{Note added} As we were completing this paper, we learned of concurrent work on a closely related subject by Abhijit Gadde and Sriram Akella \cite{akella2026renyimultientropygraphstates}. We thank these authors for synchronizing the arXiv submissions of our papers.

\providecommand{\tq}{{\mathtt q}}

\providecommand{\MIexistingfigure}[2]{%
  \IfFileExists{figs/GHZWMEc.pdf}{%
    \includegraphics[width=#2\textwidth,page=#1]{figs/GHZWMEc.pdf}%
  }{%
    \fbox{\parbox{0.80\textwidth}{\centering
      \textcolor{red}{Original artwork: \texttt{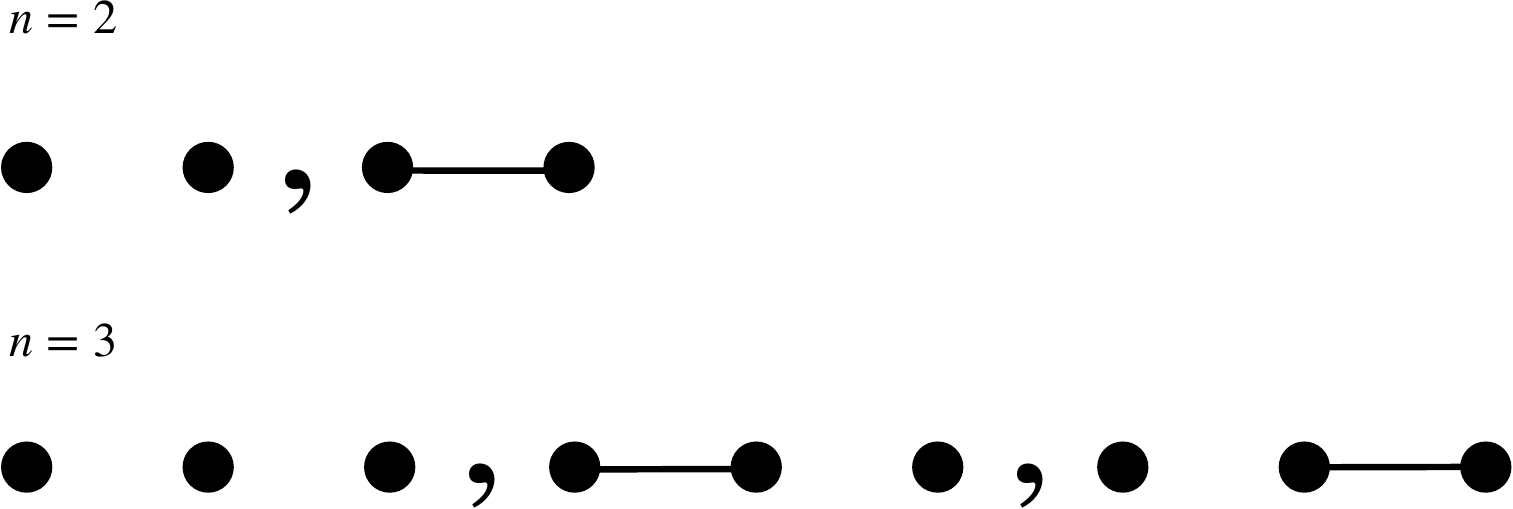},
      page #1. Insert the existing figure here; no replacement
      artwork has been generated.}}}%
  }%
}

\section{Multi-invariants of interest}
\label{sec:multis}

We consider multi-invariants of $\tq$-partite pure states. In the Introduction, we described them as complete contractions of wavefunctions
$|\Psi\rangle_{A_1A_2\ldots A_{\tq}}$ and
${}_{A_1A_2\ldots A_{\tq}}\langle\Psi|$, treated like pieces of a jigsaw
puzzle. In this paper, we focus on connected patterns whose symmetries
allow them to be organized as regular tessellations of orientable Riemann surfaces. The original motivation for this class of quantities is that it generalizes the replica construction of R\'enyi entropies in two-dimensional conformal field theories~\cite{2025JHEP...06..116G}. We shall see, however, that these patterns are useful in general quantum mechanical systems, without reference to conformal symmetry.

\subsection{Symmetry groups of jigsaw patterns}
\label{subsec:replica-contractions}

Represent wavefunctions $|\Psi\rangle_{A_1A_2\ldots A_{\tq}}$ as $\tq$-gons, with sides labeled by the $\tq$ constituent subsystems. Going clockwise around each $\tq$-gon, the sides are labeled $A_1$, $A_2,$ $\ldots,$ $ A_{\tq}$. The index is to be understood cyclically so that $A_{\tq} \equiv A_0$. In a similar fashion, the wavefunction ${}_{A_1A_2\ldots A_{\tq}}\langle\Psi|$ is also a $\tq$-gon, except the ordering $A_1$, $A_2,$ $\ldots,$ $ A_{\tq}$ circumnavigates the $\tq$-gon in the anticlockwise direction. Our object of study is a quantity obtained by gluing $s$ copies of ket $\tq$-gons with $s$ copies of bra $\tq$-gons, where gluing represents taking the partial trace. All gluings we consider satisfy two conditions:\footnote{In general, one might consider lifting some or all of these conditions, but this requires some arbitrary isomorphisms between the Hilbert spaces $\mathcal{H}_{A_i}$ and their duals.}
\begin{itemize}
\item[(i)] Kets are only glued to bras.
\item[(ii)] Only sides labeled by the same system $A_a$ can glue up.
\end{itemize}

Let us do the gluing in two stages: First glue up the sides labeled $A_q$, then glue all the others. After the first stage, we obtain $s$ copies of the ($\tq-1$)-partite reduced density matrix 
\begin{equation}
  \rho=\operatorname{Tr}_{A_q}|\Psi\rangle\langle\Psi|,
  \label{eq:reference-density-matrix}
\end{equation}
which is represented as a ($2\tq-2$)-gon with $\tq-1$ labeled ket sides and $\tq-1$ bra sides. The subsequent pattern according to which the $s$ density matrices are contracted defines $\tq-1$ permutations $g_a \in S_s$ ($1 \leq a \leq \tq-1$). The permutations are understood as follows: the $A_a$-labeled bra-type edge in the $r^{\rm th}$ copy of $\rho$ is connected to the $A_a$-labeled ket-type edge in the $g_a(r)^{\rm th}$ copy of $\rho$. With this convention, substituting the wavefunction
\begin{equation}
  |\Psi\rangle
  =\sum_{\{i_a\}}
    \psi_{i_1\, i_2 \ldots\, i_{\tq}}
    |i_1\rangle\, |i_2\rangle \ldots |i_{\tq}\rangle
  \label{wav-def}
\end{equation}
in the fully contracted jigsaw puzzle evaluates to:
\begin{equation}
  \mathcal E(g_1,\ldots,g_{\tq-1},e;\Psi)
  =
  \prod_{r=1}^{s}
  \sum_{\{i_a^{(r)}\}}
    \psi_{i_1^{(r)}\ldots\, i_{\tq-1}^{(r)} i_{\tq}^{(r)}}\,
    \overline{
      \psi_{i_1^{(g_1(r))}\ldots\, i_{{\tq-1}}^{(g_{\tq-1}(r))} i_{\tq}^{(r)}}
    }.
  \label{inv-def}
\end{equation}
Note that the partial trace over $A_\tq$ in (\ref{eq:reference-density-matrix}) effectively sets the permutation $g_\tq$ to the identity. Equation~(\ref{inv-def}) and the discussion leading to it accommodate an entirely general pattern of index contractions. We presently specialize to narrower settings.

Under additional assumptions (iii-iv) below, the group generated by the pattern of contractions $G = \langle g_1, g_2 \ldots g_{\tq-1} \rangle$ will be referred to as {\bf the replica group}. A special case with $\tq=2$ is $\mathbb{Z}_n = \langle (1\, 2 \ldots n) \rangle$, which defines the contraction pattern (\ref{eq:moments}) in the R{\'e}nyi entropies. We have constructed $G$ as a subgroup of $S_s$, but this imposes no restrictions on $G$ because every finite group can be embedded in some permutation group. 

We impose two additional conditions on $G$:
\begin{itemize}
\item[(iii)] Transitivity: The $s$ copies of $\rho$ form a single orbit under the action of $G$. Equivalently, the jigsaw puzzle we have assembled is connected. 
\item[(iv)] Free action: No $g \in G$ except the identity fixes any polygon.
\end{itemize}
Condition~(iv) can be equivalently restated in another, useful way:
\begin{itemize}
\item[(iv')] As permutations in $S_s$, all elements of $G$ are composed of disjoint of cycles of equal length.
\end{itemize}
Conditions~(iii-iv) define {\bf a regular action} of $G$ on the $s$ copies of $\rho$. These assumptions guarantee a special relationship between $G$ and the jigsaw pattern:
\begin{itemize}
    \item Choose one arbitrary ($2\tq-2$)-gon in the jigsaw pattern and label it with the identity $e \in G$. This induces a unique bijection between ($2\tq-2$)-gons and group elements $g \in G$, in the sense that a path, which starts at polygon $e$ and crosses edges labeled $A_{a_1}, A_{a_2} \ldots A_{a_k}$, ends at the polygon labeled $g = g_{a_1} g_{a_2} \ldots g_{a_k}$. This assignment of ($g \in G$)-labels to ($2\tq-2$)-gons is well defined, i.e. path-independent.
    \item Equivalently, let us collapse the jigsaw pattern to its skeleton graph: the ($2\tq-2$)-gons become vertices, which are connected by $A_a$-labeled edges. The resulting graph is the Cayley graph of group $G$.
    \item The replica group $G$ is isomorphic to the automorphism group of the jigsaw.
\end{itemize}
The bijection between copies of $\rho$ and group elements $g$ justifies why we call $G$ the replica symmetry group.

\paragraph{Extended replica symmetry} In this text, we shall also need {\bf the extended replica group} $\hat{G} = \langle r_0, g_1, g_2 \ldots g_{\tq-1} \rangle$. The additional generator $r_0$ reflects the ($2\tq-2$)-gon labeled $e$ across its $A_\tq$-labeled edge, which is contracted in (\ref{eq:reference-density-matrix}). (Recall that $A_0 \equiv A_{\tq}$.) This characterization uniquely defines the action of $r_0$ on the rest of the jigsaw. Generator $r_0$ exchanges bras $\leftrightarrow$ kets and therefore induces orientation-switching, extended automorphisms of the jigsaw pattern. Group $\hat{G}$ has order $2s$.

In analogy to how $r_0$ is defined, we may set $r_a$ to be the reflection across the $A_a$-labeled edge of the $e$-labeled ket. The generators $g_a$ of $G$ satisfy
\begin{equation}\label{eq:ga-r0ra}
g_a = r_0 r_a \qquad (1 \leq a \leq \tq -1).
\end{equation}
Thus, we may also write $\hat{G} = \langle r_0, r_1, r_2 \ldots r_{\tq-1} \rangle$. Clearly, $G$ is the normal subgroup of $\hat{G}$, which comprises elements generated by an even number of $r$'s. 

If we think of every $\rho$ in (\ref{eq:reference-density-matrix}) as two glued $\tq$-gons---a bra and a ket---then the jigsaw pattern consists of $s$ ket-type and $s$ bra-type $\tq$-gons. Similar to the relationship between the jigsaw and group $G$, we may label the $2s$ constituent $\tq$-gons with elements of $\hat{G}$. The skeleton of the jigsaw of $2s$ constituent $\tq$-gons is, once again, the Cayley graph of $\hat{G}$. Group $\hat{G}$ is called the extended replica symmetry because its elements label all wavefunctions $|\Psi\rangle_{A_1A_2\ldots A_{\tq}}$ and ${}_{A_1A_2\ldots A_{\tq}}\langle\Psi|$, which are contracted in quantity~(\ref{inv-def}).

\subsection{Meaning in conformal field theory}
\label{subsec:replica-geometry}
In two-dimensional conformal field theory, the machinery described above takes on an additional meaning. We briefly sketch this connection, referring the interested readers to the original references \cite{2025JHEP...06..116G,HarperTakayanagiTsuda2024}. Readers who only care about finite-dimensional Hilbert spaces can skip this section at no peril.

Suppose that wavefunction $|\Psi\rangle_{A_1A_2\ldots A_{\tq}}$ is prepared by a path integral on the Euclidean half-plane, possibly with operator insertions. If we map the half-plane conformally to a $\tq$-gon, the resulting quantum state is still $|\Psi\rangle_{A_1A_2\ldots A_{\tq}}$, up to a change in normalization. Our jigsaw puzzle glued from $s$ ket-type and $s$ bra-type $\tq$-gons is, therefore, an $s$-sheeted cover of the Euclidean plane. The evaluation of (\ref{eq:reference-density-matrix}) in two-dimensional conformal field theory boils down to the path integral on a Riemann surface $\Sigma$, which covers the plane $s$ times.

Many properties of $\Sigma$ can be easily computed in closed form. By the Riemann-Hurwitz formula, its genus $g_\Sigma$ is given by
\begin{equation}
  2-2g_\Sigma
  =s\left(2-\tq+\sum_{a=0}^{\tq-1}\frac{1}{\ell_a}\right),
  \label{eq:RH}
\end{equation}
where $\ell_a$ is the order of the group element $r_a r_{a+1} = g_a^{-1} g_{a+1}$. By property (iv') above, every such element acts on the $2s$ $\tq$-gons in $s / \ell_a$ orbits of length $\ell_a$. 

In uniformizing coordinates on $\Sigma$, the action of $r_a r_{a+1} = g_a^{-1} g_{a+1}$ becomes a rotation by $2 \pi / \ell_a$ around each image of the point that separates $A_a$ from $A_{a+1}$; $\Sigma$ contains $s / \ell_a$ such images. Therefore---still in uniformizing coordinates---the interior angle of the ket $\tq$-gon at the $A_a / A_{a+1}$-separating vertex is $\pi / \ell_a$. Because $G$ and $\hat{G}$ act transitively on the $s$ identical polygons in $\Sigma$, we can quotient $\Sigma$ by both these groups. Then the $(2\tq-2)$-gon that prepares $\rho$ is a fundamental domain of $\Sigma / G$ whereas the $\tq$-gon that prepares $|\Psi\rangle_{A_1A_2\ldots A_{\tq}}$ is a fundamental domain of $\Sigma / \hat{G}$.

\subsection{Examples}
\label{subsec:multi-invariant-examples}
We give three examples. The first---the three-party dihedral invariant---serves as a pedagogical warm-up; it will be used in Section~\ref{sec:warm-up} to illustrate our coloring method. The other two examples will be our main testing grounds for the non-GHZ calculations in Section~\ref{sec:coloringmethod}.

Occasionally, we write group presentations in the form $\langle\text{generators}\mid\text{relations}\rangle$, with an abstract set of generators and relations. This differs slightly from the notational convention $G = \langle g_1, g_2 \ldots g_{\tq-1} \rangle$ employed thus far, where the $g_a$ were drawn from the permutation group and so their relations were set implicitly by the multiplication table of $S_s$.

\paragraph{Three-party dihedral invariants} were introduced
in Ref.~\cite{2025JHEP...06..116G} and studied further in
Ref.~\cite{HarperMollabashiTakayanagiTasuki2025}. We take $s=2n$ copies of the state $|\Psi\rangle_{OAB}$ labeled $1,2 \ldots 2n$ and $2n$ copies of ${}_{OAB}\langle \Psi|$ labeled $\overline{1}, \overline{2}\ldots \overline{2n}$. The replica symmetry group $G\simeq \mathbb{D}_{2n}$\footnote{Throughout this paper, the subscript under $\mathbb{D}$ denotes the order of the dihedral group.} is generated by the permutations:
\begin{equation}
  g_A=(2n\,\, 2n\!-\!1\, \ldots\, n\!+\!1) (n\,\,n\!-\!1\, \ldots\, 1),
  \qquad
  g_B=(2n\,\, 1)(2n\!-\!1\,\, 2)\,\ldots\, (n\!+\!1\,\, n).
  \label{eq:dihedral-three-matchings}
\end{equation}
The extended replica group $\hat{G}$ is
\begin{equation}
  \hat G=\left\langle r_O,r_A,r_B\ \middle|\
  r_O^2=r_A^2=r_B^2=(r_Or_A)^n=(r_Ar_B)^2=(r_Br_O)^2=e
  \right\rangle,
  \label{eq:dihedral-extended-group}
\end{equation}
where we introduce reflections $r_O,r_A,r_B$ across the three edges of the reference ket labeled $e$ and use $g_A=r_O r_A$ and $g_B=r_O r_B$ from~\eqref{eq:ga-r0ra}.

In two-dimensional conformal field theory, the last three relations mean that the path integral preparation of $|\Psi\rangle_{OAB}$ is done on a triangle with interior angles $\pi / n$ (at the $O/A$ corner) and $\pi/2$ (at the other two corners), assuming uniformized coordinates. With these data
equation~(\ref{eq:RH}) gives $g_\Sigma = 0$ so the jigsaw pattern is topologically a sphere. 

This sphere is clearly visible in Figure~\ref{fig:intro3partyjigsaw} (in the case $n=3$). In Figure~\ref{fig:all1}, we redraw this figure so as to make explicit the fact that $G$ and $\hat{G}$ act as automorphism groups of the assembled Riemann surface. We do so by labeling the constituent polygons of $\Sigma$ with replica indices $1,2 \ldots 6$ and $\overline{1}, \overline{2}\ldots \overline{6}$ and then again with group elements of $G$ and $\hat{G}$.

\begin{figure}[t!]
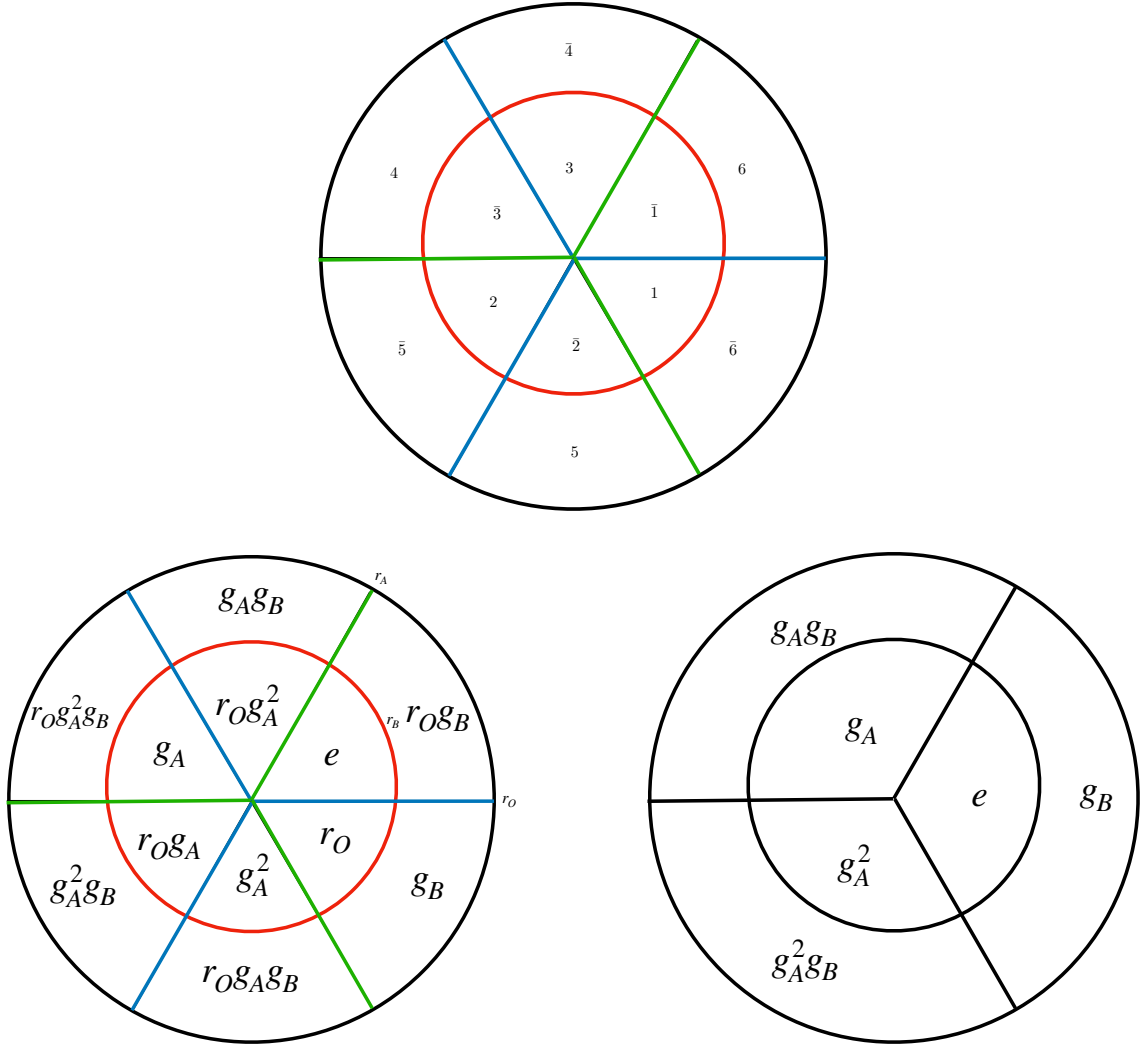

    \centering

    \par
      \MIexistingfigure{36}{0.450}

    \vspace{0.5cm}

  \MIexistingfigure{35}{0.450}
    \hfill
     \MIexistingfigure{43}{0.435}

    \caption{\textbf{Top}: The dihedral tessellation for $n=3$ with replica labels;
  unbarred and barred labels indicate ket and bra pieces. consisting of twelve
  elementary triangles with angles $(\pi/2,\pi/2,\pi/3)$. \textbf{Bottom left}: The same tessellation labeled by group elements of $\hat{G}$ acting on a reference ket. \textbf{Bottom right}: The coarsened tessellation by $(2\tq-2)$-gons (here 4-gons) that prepare reduced density matrices $\rho_{AB}$, labeled by group elements of $G$. }
    \label{fig:all1}
\end{figure}

\paragraph{R\'enyi multi-entropy.}

The $\tq$-party R\'enyi-$n$ multi-entropy uses a cyclic replica
direction for each non-reference party~\cite{GaddeKrishnaSharma2022}.
Choose $O=A_q$\footnote{It is obvious that the multi-entropy is symmetric under permutations of the parties, taking $A_q=O$ as the reference party here is merely a state-independent convention.} and label the replicas by
\begin{equation}
  \mathbf x=(x_1,\ldots,x_{\tq-1})\in(\mathbb Z_n)^{\tq-1},
  \qquad s=n^{\tq-1},
  \qquad n\geq2,
  \label{eq:renyi-replica-lattice}
\end{equation}
where $\mathbb Z_n$ denotes integers modulo $n$. For $a=1,\ldots,\tq-1$,
let $\mathbf e_a$ be the unit vector in the $x_a$ direction. The
gluing permutations are
\begin{equation}
  g_1(\mathbf x)=\mathbf x,
  \qquad
  g_a(\mathbf x)=\mathbf x+\mathbf e_a
  \quad(a=1,\ldots,\tq-1).
  \label{eq:renyi-gluing}
\end{equation}
Thus the $A_a$ index of the bra at $\mathbf x$ meets the $A_a$ index
of the ket at $\mathbf x+\mathbf e_a$. The $\tq-1$ independent shifts
commute and generate $G\simeq\mathbb Z_n^{\tq-1}$. The extended group
also includes a ket--bra reflection, which conjugates each of these
translations to its inverse; its order is $2n^{\tq-1}$.

The coordinate array in \eqref{eq:renyi-replica-lattice} is the
\emph{replica lattice} used in our later assignment calculations;
its $\tq-1$ cyclic directions label replicas. We emphasize that the $(\tq-1)$-dimensional replica lattice is not the glued jigsaw pattern $\Sigma$, which is two-dimensional.

Denoting the polynomial contraction with $\mathcal E_n^{(\tq)}$, the R\'enyi-$n$ multi-entropy in the convention of \cite{HarperTakayanagiTsuda2024} is:
\begin{equation}
  S_n^{(\tq)}(A_1:\cdots:A_{\tq})
  =\frac{1}{(1-n)n^{\tq-2}}
     \log\mathcal E_n^{(\tq)}.
  \label{eq:renyi-multientropy-definition}
\end{equation}

The factor $n^{-(\tq-2)}$ ensures that adjoining an
uncorrelated pure party leaves the entropy unchanged; its additional
replica direction produces $n$ independent copies of the original
contraction. or $\tq=2$, the construction gives
$\mathcal E_n^{(2)}=\operatorname{Tr}\rho_{A_2}^{n}$ and
\eqref{eq:renyi-multientropy-definition} is the ordinary R\'enyi
entropy. The \emph{multi-entropy} is
\begin{equation}
  S^{(\tq)}=\lim_{n\to1}S_n^{(\tq)}
  \label{eq:multientropy-limit}
\end{equation}
whenever an analytic continuation has been specified and the limit
exists.

\begin{figure}[t]
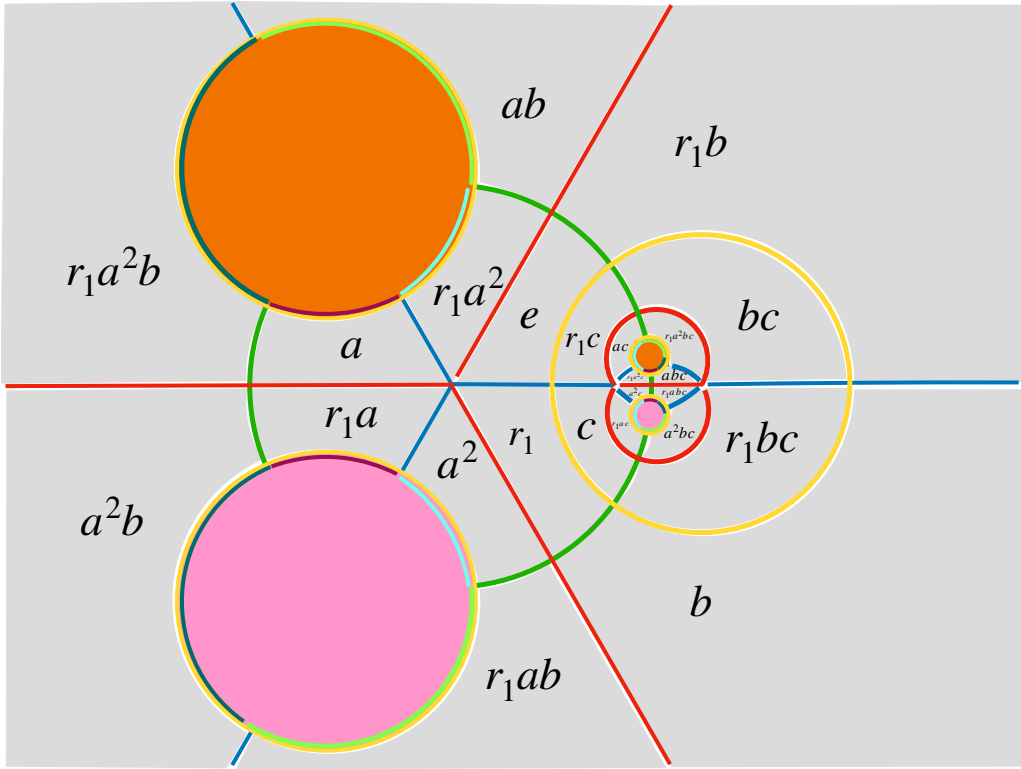

  \centering
  \MIexistingfigure{39}{0.90}
  \caption{The four-party genus-two contraction with twenty-four
  elementary quadrilaterals. Its group and party assignments are
  given in
  \eqref{eq:genus-two-generator-convention}--\eqref{eq:genus-two-gluing}.
  }
  \label{fig:4partyg2replicas}
\end{figure}

\paragraph{A four-party genus-two invariant} was first considered
in Ref.~\cite{2025JHEP...06..116G}. 
We take $s=12$ copies of the state $|\Psi\rangle_{A_1A_2A_3A_4}$, with
cyclic side order $(A_1,A_2,A_3,A_4)$ and $O=A_1$, which will remain fixed in the later cluster-state calculations. As in the previous example of three-party dihedral invariants, we introduce reflections 
$r_1,r_2,r_3,r_4$ across the four edges of the reference ket labeled
$e$ and set:
\begin{equation}
  a=r_2r_1,\qquad b=r_3r_1,\qquad c=r_4r_1
  \label{eq:genus-two-generator-convention}
\end{equation}
The extended group is:
\begin{equation}
  \begin{split}
    \hat G=\langle r_1,a,b,c\mid\;&
      a^3=b^2=c^2=r_1^2=e,\\
      &(ab)^2=(ac)^2=(bc)^2=e,\\
      &r_1ar_1=a^{-1},\quad
       r_1br_1=b,\quad r_1cr_1=c\rangle.
  \end{split}
  \label{eq:genus-two-extended-group}
\end{equation}
Its even subgroup, the replica symmetry group, is:
\begin{equation}
  G=\langle a,b,c\mid
    a^3=b^2=c^2=(ab)^2=(ac)^2=(bc)^2=e\rangle
  \simeq\mathbb D_6\times\mathbb Z_2
  \label{eq:genus-two-replica-group}
\end{equation}
Here $a,b$ generate $\mathbb D_6$ and $bc$ is an independent
commuting involution. 

The orders of elements $r_a r_{a+1}$ are respectively $\{3,2,2,2\}$, so the surface $\Sigma$ is assembled from 24 quadrilaterals with interior angles $(\pi/3,\pi/2,\pi/2,\pi/2)$; see Figure~\ref{fig:4partyg2replicas}. From formula~(\ref{eq:RH}), this surface as genus $g = 2$, which explains the name of this multi-invariant.

In particular, reflections across the $A_1$- and $A_2$-labeled sides generated a subroup $\langle r_1,r_2\rangle=\langle r_1,a\rangle$ of $\tilde{G}$/. Similarly, reflections across the $A_3$- and $A_4$-labeled sides generate $\langle r_3,r_4\rangle=\langle br_1,cr_1\rangle$. These subgroups will play an important role in Section~\ref{sec:coloringmethod}. 


\section{Warm-up examples and a coloring method}\label{sec:warm-up}
As a warm-up, we consider the calculations of multi-invariants of several commonly encountered quantum states. We begin with the three-party GHZ state of qudits:
\begin{equation}
    |{\rm GHZ}\rangle = 
    \frac{1}{\sqrt{d}}\sum_{a,b,c=0}^{d-1}\delta_{ab}\delta_{bc}|a,b,c\rangle
\end{equation}
As an example, the R\'enyi-2 multi-entropy involves the following computation:
\be
\begin{split}
&\psi_{abc}\psi_{def}\psi_{ghi}\psi_{jkl}\bar\psi^{dhc}\bar\psi^{akf}\bar\psi^{jbi}\bar\psi^{gel}\\
&=\frac{1}{d^4}\sum_{a\dots x}\underbrace{\delta_{ab}\delta_{bc}\delta_{de}\delta_{ef}\delta_{gh}\delta_{hi}\delta_{jk}\delta_{kl}}_{\psi's}\underbrace{\delta_{mn}\delta_{no}\delta_{pq}\delta_{qr}\delta_{st}\delta_{tu}\delta_{vw}\delta_{wx}}_{\bar\psi's}\underbrace{\delta_{ap}\delta_{bt}\delta_{co}\delta_{dm}\delta_{ew}\delta_{fr}\delta_{gv}\delta_{hn}\delta_{iu}\delta_{js}\delta_{kq}\delta_{lx}}_{\text{contractions}}
\end{split}
\label{ghzcrazy}
\ee
The final answer is $d^{-3}$. In principle, the sum in (\ref{ghzcrazy}) involves a summation over $d^{24}$ terms---$d^3$ from each of the eight copies of $|{\rm GHZ}\rangle$ and $\langle{\rm GHZ}|$. In the end, however, only $d$ of them contribute. This is because the wavefunction $|{\rm GHZ}\rangle$, written as a vector in the $d^3$-dimensional Hilbert space, is sparse.

Our aim is to study the computation of multi-invariants in other sparse wavefunctions (with a smaller degree of sparsity). Toward that goal, it is useful to associate non-vanishing terms in contractions such as (\ref{ghzcrazy}) to colorings of the jigsaw pattern. {\bf One color will be associated with a single component of the superposition.} Thus, there are $d$ colors in the GHZ wavefunction of qudits.  A non-vanishing contribution to the contraction will be represented by a {\bf valid coloring} of the jigsaw pattern. Whether or not a general term in the summation vanishes will be represented by a distinction between valid and invalid colorings. In this language, the key to computing multi-invariants lies in distinguishing invalid colorings from valid colorings and then in counting the latter.

In the case of the three-party GHZ case on qubits ($d=2$), the coloring rules are given in the following table:

\begin{figure}[H]
  \centering
   \includegraphics[width=.6\textwidth,page=34]{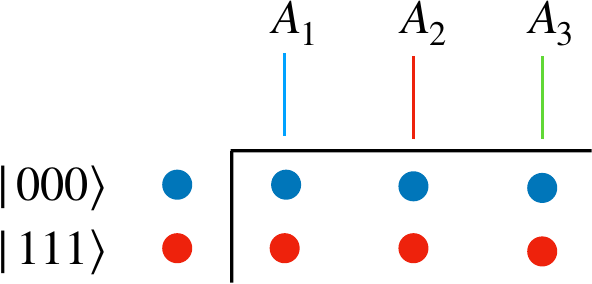}
  \caption{The coloring rules for the $\tq=3$ GHZ state of qubits ($d=2$).}
\label{fig:GHZcolorrules}
\end{figure}
\noindent 
This table should be read in the following way: If a triangle representing $|{\rm GHZ}\rangle$ is colored blue ($\psi^{000}$) then all its bra-neighbors must also be blue ($\bar\psi^{000}$), and likewise for the color red ($\psi^{111}$ and $\bar\psi^{111}$). 

These rules demand that every coloring be uniform over the entire jigsaw. We illustrate the two valid colorings pertinent to the dihedral invariant of the qubit ($d=2$) GHZ state in Figure~\ref{fig:ghzconfig}. At general $d$, we always have precisely $d$ valid colorings---regardless of the contraction being computed. Including the normalization $(d^{-1/2})^{2s}$ for a pattern with $|\hat{G}|=2s$, the jigsaw pattern evaluates to $d^{1-s}$. We emphasize that this conclusion applies to {\bf all jigsaw patterns}.
\begin{figure}[H]
\begin{tabular}{cc}
\centering
\includegraphics[width=.45\textwidth,page=37]{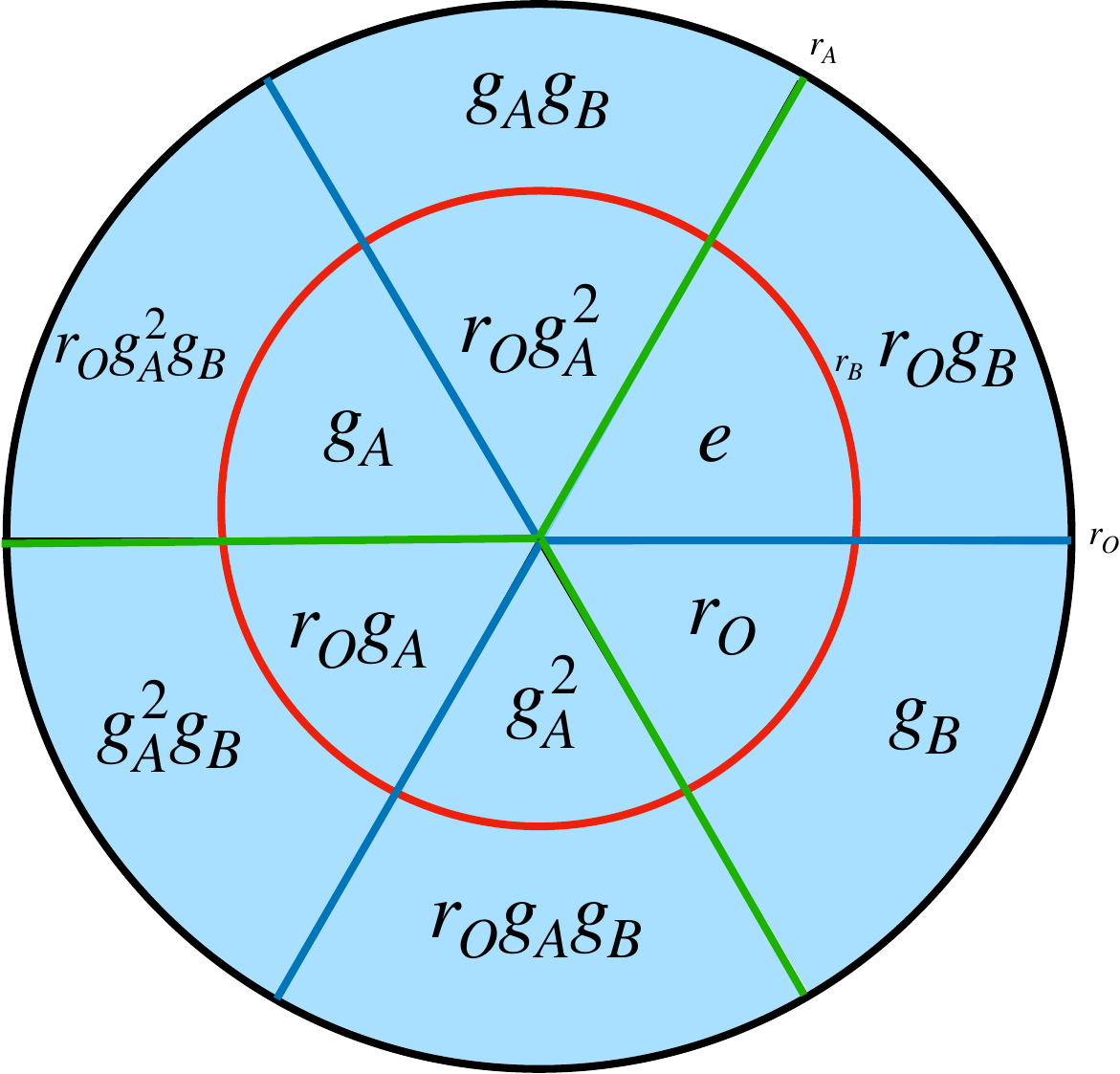}&\qquad
\includegraphics[width=.45\textwidth,page=38]{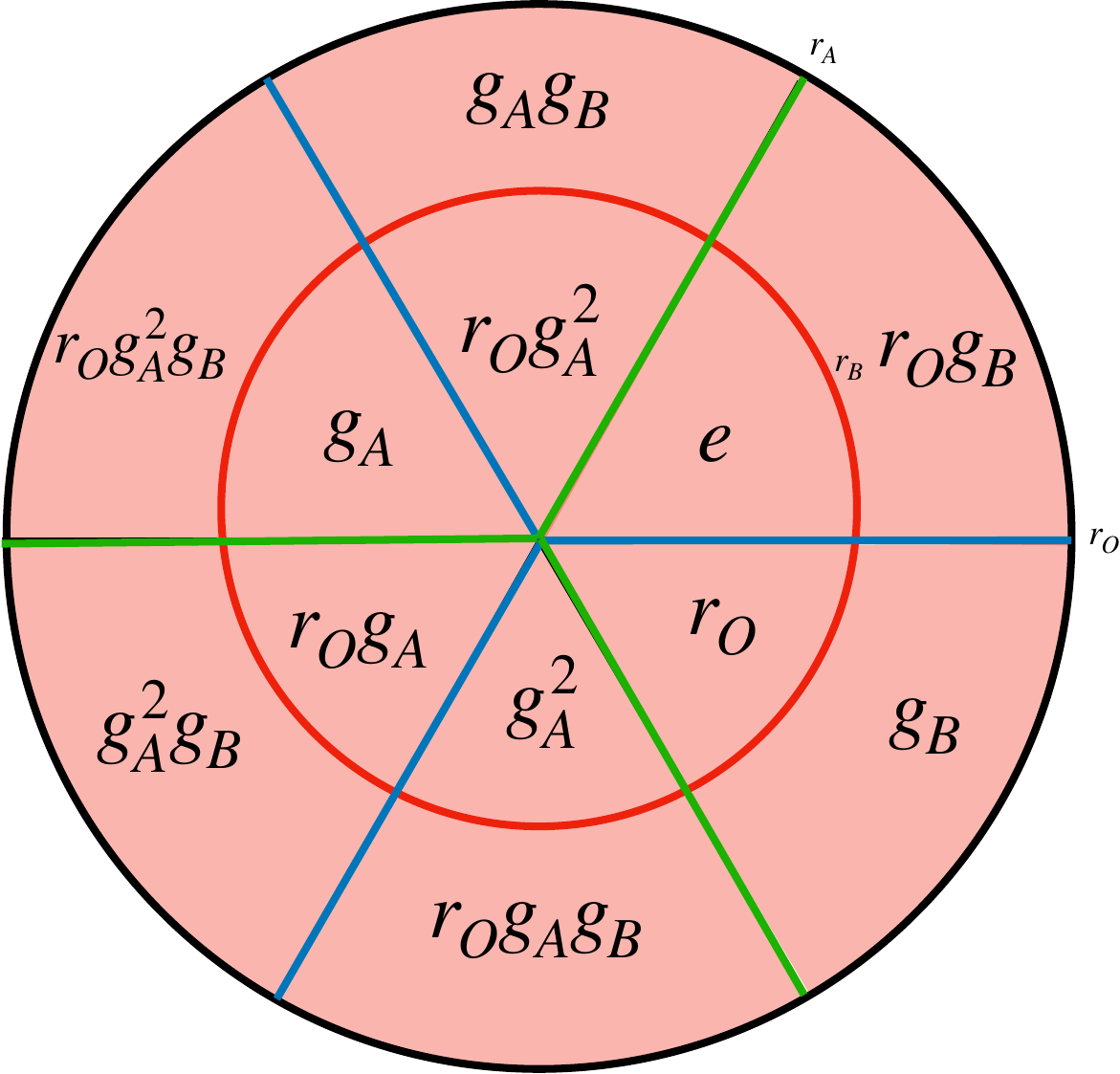}
\end{tabular}
\caption{\label{fig:ghzconfig} In the GHZ state, all valid colorings are uniform.}
\end{figure}

\paragraph{A more challenging example}
Consider the four-party R\'enyi-2 multi-entropy of the four-party cluster state:
\be\label{eq:C4}
|C_{4}\rangle=\frac{1}{2}\left(|0000\rangle+|0011\rangle+|1100\rangle-|1111\rangle\right).
\ee
Now, the $(-1)$ phase requires us to count colorings with an even and an odd number of $|1111\rangle$s and $\langle 1111|$s and then to take their difference.

As before, we start with the coloring rules:

\begin{figure}[H]
  \centering
   \includegraphics[width=.6\textwidth,page=26]{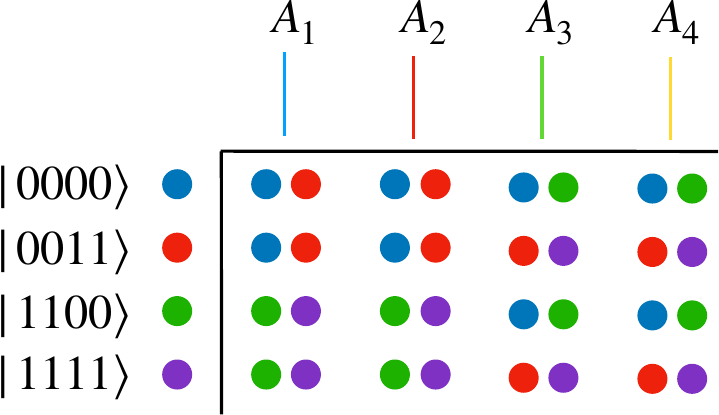}
  \caption{The coloring rules for the 4-party cluster state.}
\label{fig:4partyclustercolorrules}
\end{figure}

\noindent
Here every component of the superposition has two viable neighbors across each side. 
For a general jigsaw pattern, the task is to compute:
\begin{equation}
    2^{-(2s)} \sum_{\text{valid colorings}}(-1)^{\text{\#\{purple $\tq$-gons}\}}
    \label{validcluster4}
\end{equation}
The prefactor comes from $2s$ amplitudes of magnitude $(1/2)$, for an extended replica symmetry with $|\hat{G}|=2s$. Each valid coloring contributes $\pm 1$ depending on the number of $|1111\rangle$ and $\langle 1111|$ components.

An example valid coloring, which is relevant to computing a $G = \mathbb{Z}_2^3$ multi-invariant, is shown in Figure~\ref{fig:4partyrenyi2config}. In general, however, counting colorings by hand quickly becomes intractable. We would like to have a more principle approach. This is the subject of the next section.

\begin{figure}[H]
  \centering
   \includegraphics[width=.6\textwidth,page=27]{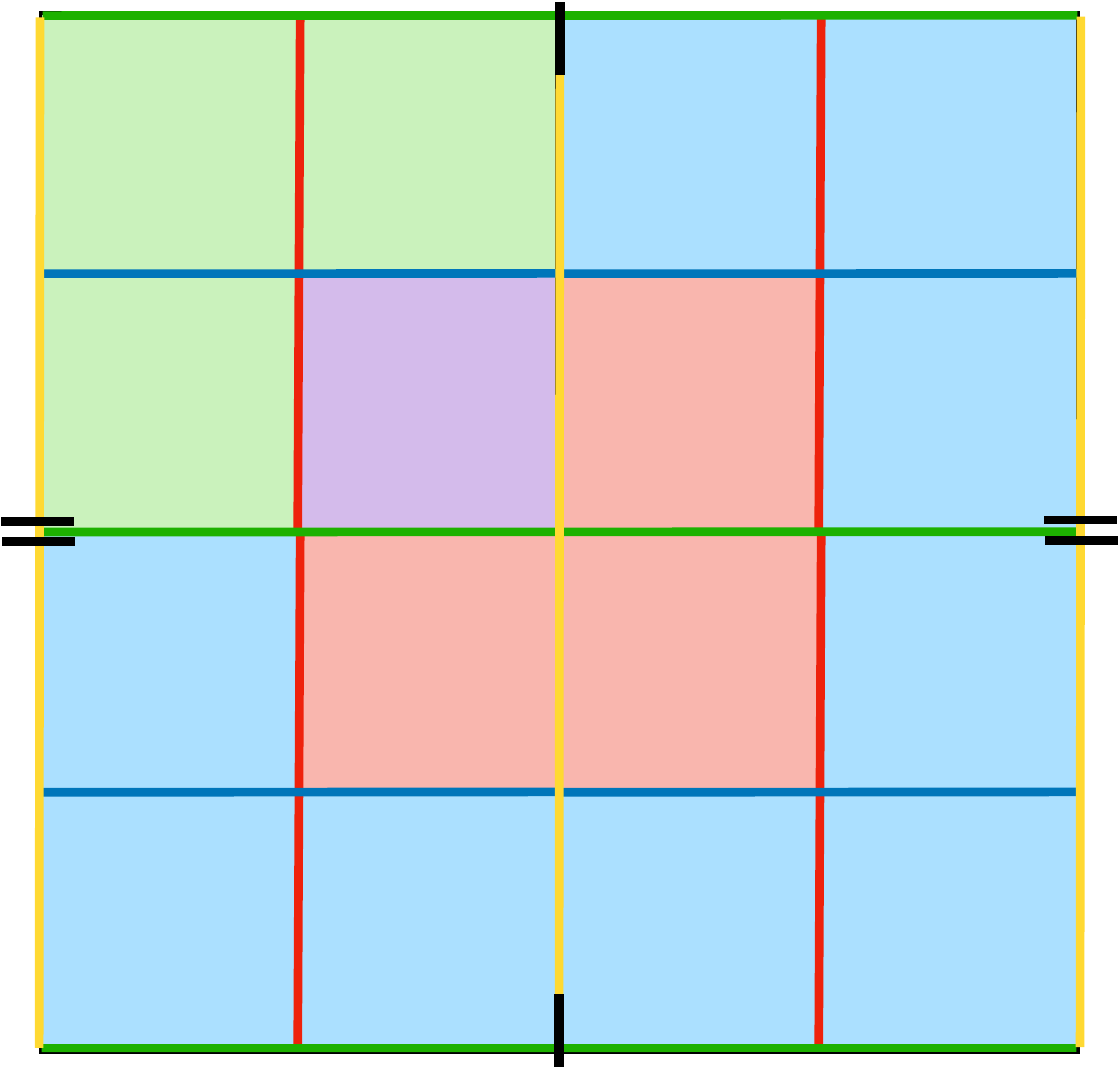}
  \caption{A valid coloring, which contributes negatively to (\ref{validcluster4}).}
\label{fig:4partyrenyi2config}
\end{figure}

\section{A group theoretic method}
\label{sec:coloringmethod}
We concentrate on state~(\ref{eq:C4}) and make the following observation. The coloring rules in Figure~\ref{fig:4partyclustercolorrules} naturally organize themselves four $2 \times 2$ bichromatic blocks. The blocks have the following meaning:
\begin{itemize}
\item If any $4$-gon is colored red/blue then its neighbors across $A_1$- and $A_2$-labeled edges are also red/blue.
\item If any $4$-gon is colored green/purple then its neighbors across $A_1$- and $A_2$-labeled edges are also green/purple.
\item If any $4$-gon is colored green/blue then its neighbors across $A_3$- and $A_4$-labeled edges are also green/blue.
\item If any $4$-gon is colored red/purple then its neighbors across $A_3$- and $A_4$-labeled edges are also red/purple.
\end{itemize}
We highlight the four $2\times 2$ blocks in Figure~\ref{fig:cluster-coloring-blocks}.

\tikzset{
  MI diagram/.style={
    x=1cm,
    y=1cm,
    font=\large
  },
  MI block zero/.style={
    draw=black,
    line width=1.6pt,
    line join=miter
  },
  MI block one/.style={
    draw=black!50,
    line width=1.6pt,
    line join=miter
  }
}

\begin{figure}[htbp]
\centering
\makebox[\linewidth][c]{%
\resizebox{0.6\paperwidth}{!}{%
\begin{tikzpicture}[MI diagram]

  \ClusterColoringBase


  \draw[MI block zero]
    (0.16,0.27) rectangle (2.30,-0.87);

  \draw[MI block one]
    (0.16,-0.93) rectangle (2.30,-2.07);


  \foreach \yy in {0,-1.20}{
    \draw[MI block zero]
      (2.38,{\yy+0.27}) rectangle (4.59,{\yy-0.27});
  }

  \foreach \yy in {-0.60,-1.80}{
    \draw[MI block one]
      (2.38,{\yy+0.27}) rectangle (4.59,{\yy-0.27});
  }

\end{tikzpicture}%
}%
}

\caption{The four $2\times 2$ blocks in the coloring rules for state~(\ref{eq:C4}), which are highlighted in (\ref{color2bits}).}
\label{fig:cluster-coloring-blocks}
\end{figure}
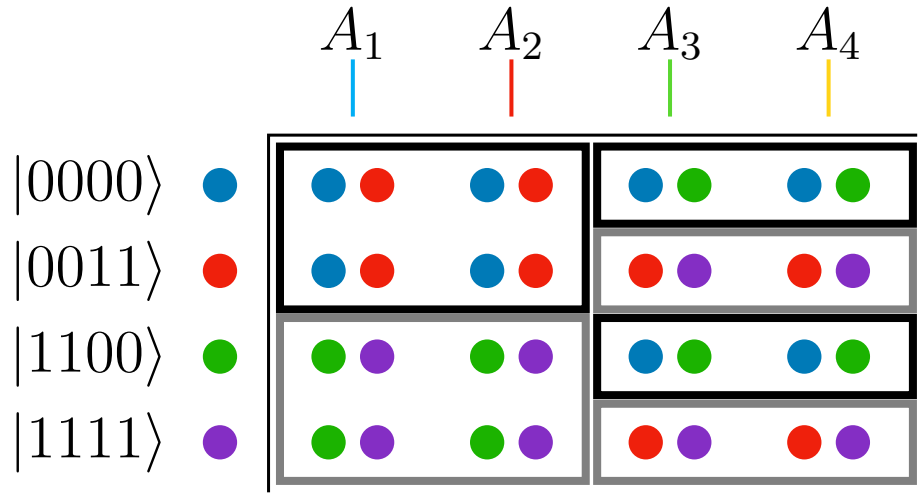

These facts suggest labeling the four colors in Figure~\ref{fig:4partyclustercolorrules} by two binary variables: red/blue (0) versus green/purple (1), and green/blue (0) versus red/purple (1). Thus, we have an identification:
\begin{align}
    {\rm blue}\, & \leftrightarrow \, 00 \nonumber \\
    {\rm red}\, & \leftrightarrow \, 01 \label{color2bits} \\
    {\rm green}\, & \leftrightarrow \, 10 \nonumber \\
    {\rm purple}\, & \leftrightarrow \, 11 \nonumber
\end{align}
This identification descends in the obvious way from the original wavefunction~(\ref{eq:C4}). Importantly for the calculation, the purple 4-gons (which determine the $\pm 1$ contribution of each valid coloring) are uniquely marked by a pair of 1's.

We further recall that, when the 4-gons in the jigsaw are labeled by elements of the extended replica symmetry $\hat{G}$, crossing $A_1$- and $A_2$-type edges amounts to applying generators $r_1$ and $r_2$. This means that a collection of 4-gons sharing the same first bit in the assignment~(\ref{color2bits}) must form cosets of the subgroup $\hat{H}_1 = \langle r_1, r_2 \rangle$ of $\hat{G}$. The same is true for the second bit, except that 4-gons with the same second bit occupy cosets of $\hat{H}_2 = \langle r_3, r_4 \rangle$. We highlight both types of cosets in Figure~\ref{fig:4partyrenyi2cosets}, using the $\mathbb{Z}_2^3$ jigsaw as an example. 

Bijection~(\ref{color2bits}) between colors and $\mathbb{F}_2^{\otimes 2}$ is powerful because it allows us to easily express necessary and sufficient conditions for a valid coloring. In particular, {\bf any two maps} $\hat{G}/\hat{H}_1 \to \mathbb{F}_2$ and $\hat{G}/\hat{H}_2 \to \mathbb{F}_2$ define a valid coloring, and every every valid coloring defines a pair of such maps. Therefore, the sum in (\ref{validcluster4}) boils down to summing over pairs of binary vectors $u \in \mathbb{F}_2^{|\hat{G}/\hat{H}_1|}$ and $v \in \mathbb{F}_2^{|\hat{G}/\hat{H}_2|}$. 

The correct weight of one term in such a such is the number of purple 4-gons. Let $a_{ij}$ be the parity of the number of elements of $\hat{G}$, which are common to the $i^{\rm th}$ coset of $\hat{H}_2$ and the $j^{\rm th}$ coset of $\hat{H}_1$:
\be
(A)_{ij}=|c_i\hat H_2\cap d_j\hat H_1|~~ (\mod{2})
\label{cosetintersect}
\ee
In this equation, we have implicitly chosen some arbitrary representatives of disjoint cosets $c_i \hat{H}_2$ and $d_j \hat{H}_1$. Matrix~$A$ defined in (\ref{cosetintersect}) is the left coset intersection matrix. We review some of its properties in Appendix~\ref{sec:leftcosetprop}.

In a valid coloring, we can count purple 4-gons by collecting 4-gons in all intersections of $\hat{H}_1$- and $\hat{H}_2$-cosets, which are labeled 11. The coloring characterized by $u \in \mathbb{F}_2^{|\hat{G}/\hat{H}_1|}$ and $v \in \mathbb{F}_2^{|\hat{G}/\hat{H}_2|}$ therefore contains $u^{\rm T} A v$ (mod 2) purple 4-gons. In other words, sum~(\ref{validcluster4}) equals:
\begin{equation}
\label{eq:arf}
    2^{-2s} 
    \sum_{u\, \in\, \mathbb{F}_2^{|\hat{G}/\hat{H}_1|}}
    \sum_{v\, \in\, \mathbb{F}_2^{|\hat{G}/\hat{H}_2|}}
    (-1)^{u^{\rm T} A v}
\end{equation}

\begin{figure}[t!]
\begin{tabular}{cc}
\centering
\includegraphics[width=.4\textwidth,page=29]{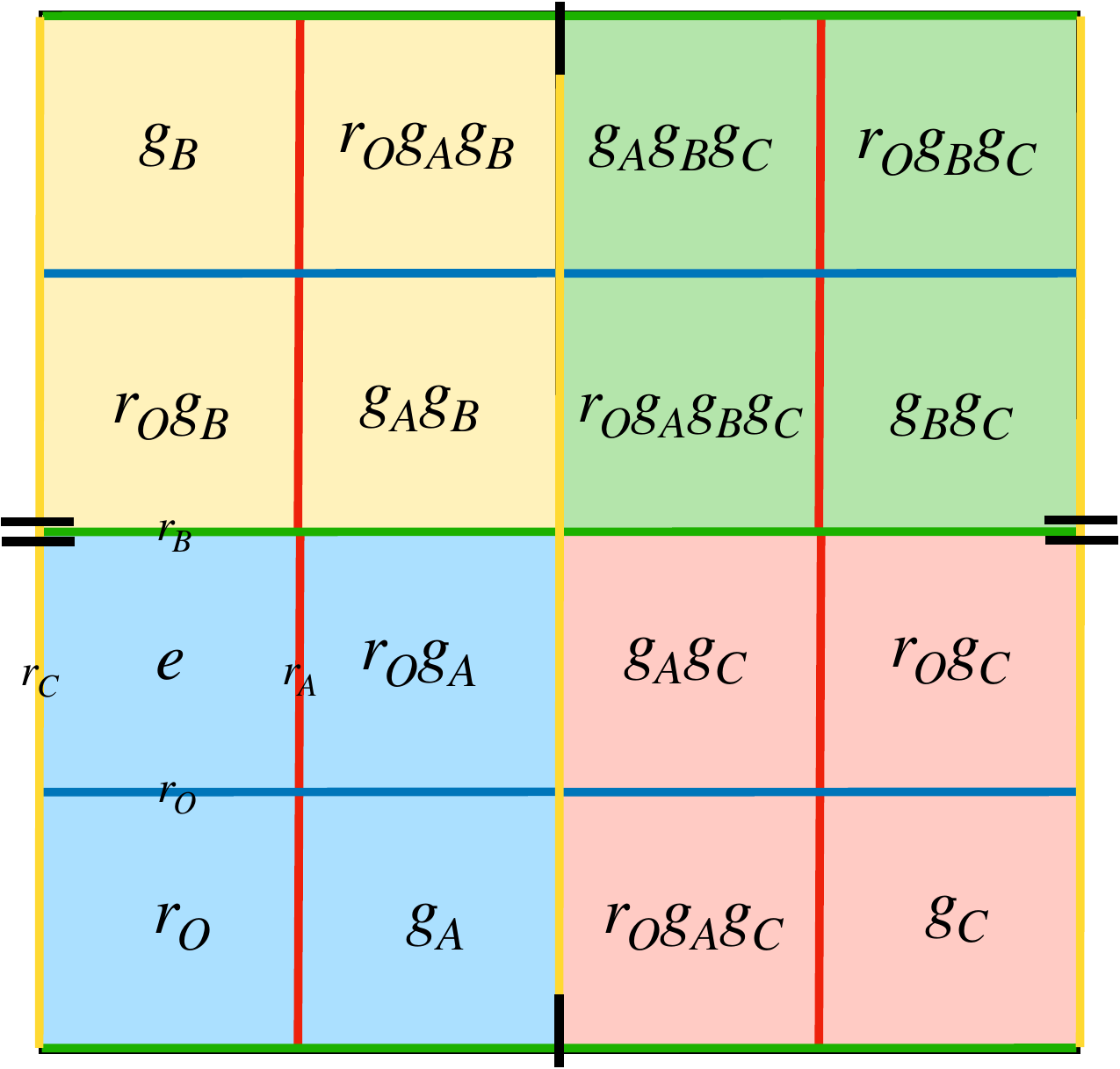}&\qquad\qquad
\includegraphics[width=.4\textwidth,page=30]{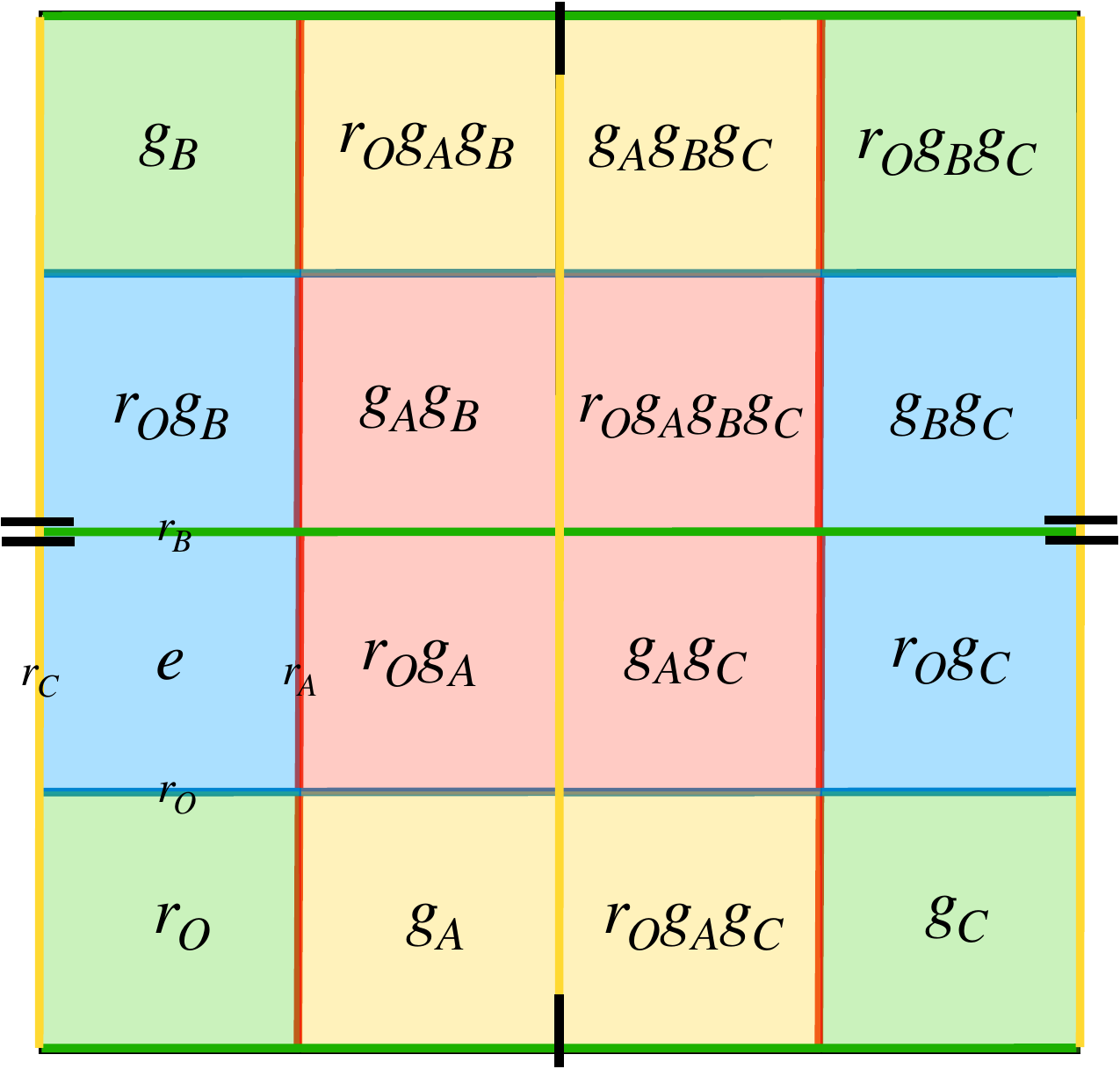}
\end{tabular}
\caption{\label{fig:4partyrenyi2cosets} A tessellation can be decomposed into left cosets of $\hat H_1=\langle r_1,r_2\rangle$ (left) and $\hat H_2=\langle r_3,r_4\rangle$ (right). Here $(O,A,B,C)\equiv(A_1,A_2,A_3,A_4)$.}
\end{figure}

 To compute the sum, define $w=A v$. Using orthogonality in $\mathbb{F}_2^{|\hat{G}/\hat{H}_1|}$:
\be
\sum_{u}(-1)^{u\cdot w}=\begin{cases} 
2^{|\hat{G}/\hat{H}_1|} & \quad w=0\\ 0 &\quad w\neq 0 \end{cases}
\ee
Thus, the problem reduces to counting the solutions of $Av=0$. From the rank-nullity theorem, there are $2^{\ker(A)}=2^{|\hat{G}/\hat{H}_2|-{\rm rank}(A)}$ of them. (Here, of course, the rank of $A$ is computed over $\mathbb{F}_2$.) The final answer becomes:
\be
 2^{-(2s)} \sum_{\text{valid colorings}}(-1)^{\text{\#\{purple $\tq$-gons\}}}
 = 
2^{-2s+|\hat{G}/\hat{H}_1|+|\hat{G}/\hat{H}_2|-{\rm rank}(A)}
\label{finalcoarse}
\ee

As an example, for the four-party R{\'e}nyi-2 multi-entropy of state~(\ref{eq:C4}) we have $2s=16$, $|\hat{G}/\hat{H}_1|=|\hat{G}/\hat{H}_2|=4$ and ${\rm rank}(A)=1$. Including the standard normalization from (\ref{eq:renyi-multientropy-definition}), we find:
\begin{equation}
  S_{n=2}^{(4)}(A_1:A_2:A_3:A_4)
  =\frac{1}{(1-2)2^{4-2}}
     \log({\rm eq.~}\ref{finalcoarse}) = \frac{9}{4}\log2
  \label{eq:renyi2ME4clusteranswer}
\end{equation}

\paragraph{Effect of the repetition code}
The four-party cluster state is the smallest example of a more general class of states, which we call core-2 states: 
\begin{equation}\label{eq:defn_core2state}
\ket{\text{core-}2_{p+q}} = \tfrac{1}{2}\Bigl(
\ket{\underbrace{0\dots 0}_p \underbrace{0\dots 0}_q}
+ \ket{\underbrace{1\dots 1}_p \underbrace{0\dots 0}_q}
+ \ket{\underbrace{0\dots 0}_p \underbrace{1\dots 1}_q}
- \ket{\underbrace{1\dots 1}_p \underbrace{1\dots 1}_q}
\Bigr)
\end{equation}
The first `core' contains $p$ parties and the second core contains $q$ parties.
All such states descend from the Bell pair $CZ|++\rangle$ via the repetition quantum code. (Here $CZ = (-1)^{(1-Z) \otimes (1-Z)/4}$ is the controlled phase gate and $|+\rangle$ is the $(+1)$-eigenstate of the Pauli $X$ operator.)

The method outlined above applies to all such states directly because the basis elements (colors) can still be labeled as in (\ref{color2bits}). What changes with $p$ and $q$ is the number of party labels, and hence the number of labeled edges in the tessellation. This change implies that the subgroups $\hat{H}_1$ and $\hat{H}_2$ have different orders:
\begin{equation}
\hat{H}_1
= \langle r_a ~~ (1 \leq a \leq p) \rangle,
\qquad
{\rm and}
\qquad
\hat{H}_2
= \langle r_b ~~ (p+1 \leq b \leq p+q) \rangle,
\label{eq:core2-subgroups-general}
\end{equation}
Formula~(\ref{finalcoarse}) applies to this case without any modifications. We emphasize that it is valid for all jigsaw-like contraction patterns. The only underlying assumption is that group $G$ acts regularly (freely and transitively) on the underlying tessellation.

\subsection{Example 1: R{\'e}nyi multi-entropy of core-2 states}
In applying equation~(\ref{finalcoarse}), we shall make use of group theoretic facts, which are reviewed in Appendix~\ref{sec:leftcosetprop}.

\paragraph{Four-party cluster state}
With four parties, the extended replica symmetry group relevant to computing the R{\'e}nyi-2 multi-entropy is: 
\be
\hat G=  \langle r_1,a,b,c|r_1^2,a^n,b^n,c^n,(r_1a)^2,(r_1b)^2,(r_1c)^2,aba^{-1}b^{-1},aca^{-1} c^{-1}, bcb^{-1}c^{-1} = e\rangle.
\ee
The other reflections are defined so $a=r_2r_1$, $b=r_3r_1$ and $c=r_4r_1$.\footnote{As in the genus-two example of Section~\ref{subsec:multi-invariant-examples}, the reference party is $A_1$, so $r_0=r_1$ and $g_1=e$, $g_2=a^{-1}$, $g_3=b^{-1}$, $g_4=c^{-1}$.} The order of the extended replica group is $|\hat G|=2n^3$. Our two subgroups are $\hat H_1=\langle r_1,r_2\rangle=\langle r_1,a\rangle$ and $\hat H_2=\langle r_3,r_4\rangle=\langle br_1,cr_1\rangle$, each of order $2n$. 

To apply Theorem~\ref{thm:atilde} from Appendix~\ref{sec:leftcosetprop}, we first determine the normal cores. (See Definition~\ref{defcore} for this terminology.) The strategy is to first identity a normal subgroup of $\hat G$ contained in $\hat H_{1,2}$ and then show that it cannot be made any larger.
\begin{itemize}
    \item $\hat H_1:$ First we note $a$ commutes with $b,c$. Therefore, we only need to check that $r_1ar_1=a^{n-1}\in\langle a \rangle $ to conclude that $\langle a\rangle$ is a normal subgroup of $\hat G$. Furthermore $\langle r_1 \rangle$ is not normal since e.g. $b^{n-1}r_1b=b^{n-2}r_1\notin\hat H_1$. So we have $C_{\hat G}(\hat H_1)=\langle a\rangle$.
    \item  $\hat H_2:$ Similarly $\langle br_1cr_1=bc^{n-1}\rangle$ is a normal subgroup of $\hat G$ since $bc^{n-1}$ commutes with $a,b,c$ and $r_1bc^{n-1}r_1=b^{n-1}c\in \langle bc^{n-1}\rangle$. But since $b^{n-1}br_1b=b^{n-1}r_1\notin\hat H_2$ we have $C_{\hat G}(\hat H_2)=\langle bc^{n-1}\rangle$.
\end{itemize}

Now we take the quotients as described in \eqref{eq:Aquotient} to find:
\be\label{eq:core2quotient}
\begin{split}
    \tilde{G}=&\,G/\langle a,bc^{n-1}\rangle=\langle b,r_1|b^n, r_1^2,(br_1)^2\rangle\sim \mathbb{D}_{2n}\\
    \tilde H_1=&\,\hat H_1/C_{\hat G}(\hat H_1)=\{e,r_1\}\sim \mathbb{Z}_2\\
     \tilde H_2=&\,\hat H_2/C_{\hat G}(\hat H_2)=\{e,br_1\}\sim \mathbb{Z}_2
\end{split}
\ee
For both $\tilde H_1$ and $\tilde H_2$ the left cosets can be indexed by the representatives $\{e,b,\cdots,b^{n-1}\}$ but for $\tilde H_2$ the second element will have exactly one additional power of $b$ (mod $n$). This means that for all $n$ the coset intersections have the simple form:
\be
b^j\tilde H_2 ~\text{ intersects }~ b^j\tilde H_1 \text{ and } b^{j+1}\tilde H_1\; (j\mod n).
\label{highnintersects}
\ee
The case $n=2$ is exceptional in that the extended replica symmetry group is Abelian:
\be
\hat G=  \langle r_1,a,b,c|r_1^2,a^2,b^2,c^2,(r_1a)^2,(r_1b)^2,(r_1c)^2,ab=ba,ac=ca,bc=cb\rangle\sim \mathbb{Z}_2^4
\ee
That is because the relations $(r_1a_i)^n=e$ at $n=2$ imply that $r_1$ commutes with all the other generators. As a consequence, both $\hat{H}_1$ and $\hat{H}_2$ are normal subgroups and therefore $\tilde G\sim \tilde H_1\sim\tilde H_2$ are all trivial. Thus, $A$ at $n=2$ is the $1 \times 1$ identity matrix.

Employing~(\ref{highnintersects}) for higher $n$, we have the following sequence of intersection matrices indexed by $n \geq 2$:
\be
\begin{pmatrix}
1 
\end{pmatrix} \qquad
\begin{pmatrix}
1 & 1 & 0  \\
0 & 1 & 1 \\
1 & 0 & 1  
\end{pmatrix}\qquad
\begin{pmatrix}
1 & 1 & 0 & 0  \\
0 & 1 & 1 & 0 \\
0 & 0 & 1 & 1\\
1 & 0 & 0 & 1
\end{pmatrix}
\label{circulants}
\ee
Their ranks equal $n-1$. We prove this fact in Appendix~\ref{app:circ}. 

Reinstating all factors from definition~(\ref{eq:renyi-multientropy-definition}) and using~(\ref{finalcoarse}), we find for the R{\'e}nyi multi-entropy of state~(\ref{eq:C4}):
\be
  S_n^{(4)}(A:B:C:D) = 
\frac{1}{(1-n)n^2}(-2n^3+2n^2-n+1)\log2=\left(2+\frac{1}{n^2}\right)\log2
\ee
This generalizes result \eqref{eq:renyi2ME4clusteranswer} to $n\geq 3$. Taking the $n\rightarrow1$ limit (with the canonical analytic continuation), we obtain the multi-entropy equal to $3\log2$.

\paragraph{Other core-2 states}
The above result can then be extended to all core-2$_{p+q}$ graph states. Throughout this paragraph, $p$ and $q$ (with $p,q \geq 2$) denote the sizes of the two cores, so that the total number of parties is $p+q$.
Now we have for the extended replica symmetry group
\be
\hat G=  \langle \{a_i\},r_1~|~a_i^n=r_1^2=(r_1a_i)^2=e,a_ia_j=a_ja_i\rangle,
\ee
where for the $(p+q)$-party state $i\in\{1,\dots,p+q-1\}$.
The two subgroups are now:
\be
\hat H_1 =\langle r_1, a_1r_1,\cdots a_{p-1}r_1\rangle, \quad \hat H_2 =\langle a_{p}r_1,\cdots a_{p+q-1}r_1\rangle
\ee
Their orders and indices are $|\hat H_1|=2n^{p-1}$, $|\hat G:\hat H_1|=n^{q}$ and $|\hat H_2|=2n^{q-1}$, $|\hat G:\hat H_2|=n^{p}$. The cores defined in Definition~\ref{defcore} are:
\be
C_{\hat G}(\hat H_1)=\langle a_1,a_2,\cdots, a_{p-1}\rangle \quad C_{\hat G}(\hat H_2)=\langle a_{p}a_{p+1}^{n-1},a_{p}a_{p+2}^{n-1},\cdots, a_{p}a_{p+q-1}^{n-1}\rangle
\ee
Importantly, when we take the quotients we get \emph{exactly} the same $\tilde G, \tilde H_1, \tilde H_2$ and, consequently, ${\rm rank}(A)=n-1$. Collecting all factors in (\ref{eq:renyi-multientropy-definition}) together, we find:
\be
S_n^{(p+q)} = 
\frac{1}{1-n}\frac{1}
{n^{p+q-2}}\log\left(\frac{2^{n^p+n^q-n+1}}{2^{2n^{p+q-1}}}\right)=\left(\sum_{k=0}^{p-2}\frac{1}{n^k}+\sum_{k=0}^{q-2}\frac{1}{n^k}+\frac{1}{n^{p+q-2}}\right)\log 2
\ee
The $n \to 1$ limit with the canonical analytic continuation gives the multi-entropy equal to $(p+q-1)\log 2$.

\subsection{Example 2: Four-party measure with genus 2}
For the definition, refer to Section~\ref{subsec:multi-invariant-examples} and Figure~\ref{fig:4partyg2replicas}.

\paragraph{Four-party cluster state}
The two subgroups are given by $\hat H_1=\langle r_1,r_2 \rangle=\langle r_1,a \rangle\sim\mathbb{D}_6$ and $\hat H_2=\langle r_3,r_4 \rangle= \langle br_1,cr_1 \rangle\sim \mathbb{Z}_2^2$. It can be confirmed that they are both normal in $\hat G$, so $C_{\hat G}(\hat H_1)=\hat H_1$ and $C_{\hat G}(\hat H_2)=\hat H_2$. Because $\hat H_1\cap\hat H_2=\{e\}$ and $|\hat H_1||\hat H_2|=24=|\hat G|$, the quotient group $\tilde G=\hat G/\langle\hat H_1,\hat H_2\rangle$ is trivial. Equivalently, the fact that $\hat G=\hat H_2e\hat H_1$ makes the intersection matrix $A_{6\times4}$ consist entirely of ones; see Figure~\ref{fig:4partyg2Hcosets}. Consequently, the rank of the intersection matrix is 1. 

Substituting $2s=24$, $|\hat G/\hat H_1|=4$, $|\hat G/\hat H_2|=6$, and ${\rm rank}(A)=1$ into~(\ref{finalcoarse}) (or directly using equation~(\ref{eq:bothnormal}) from the Appendix), the genus-2 multi-invariant evaluates to:
\be
\CE(e,a^{-1},b,c;\ket{C_4})=2^{-24+4+6-1}=2^{-15}.
\ee

\begin{figure}[H]
\begin{tabular}{cc}
\centering
\includegraphics[width=.4\textwidth,page=40]{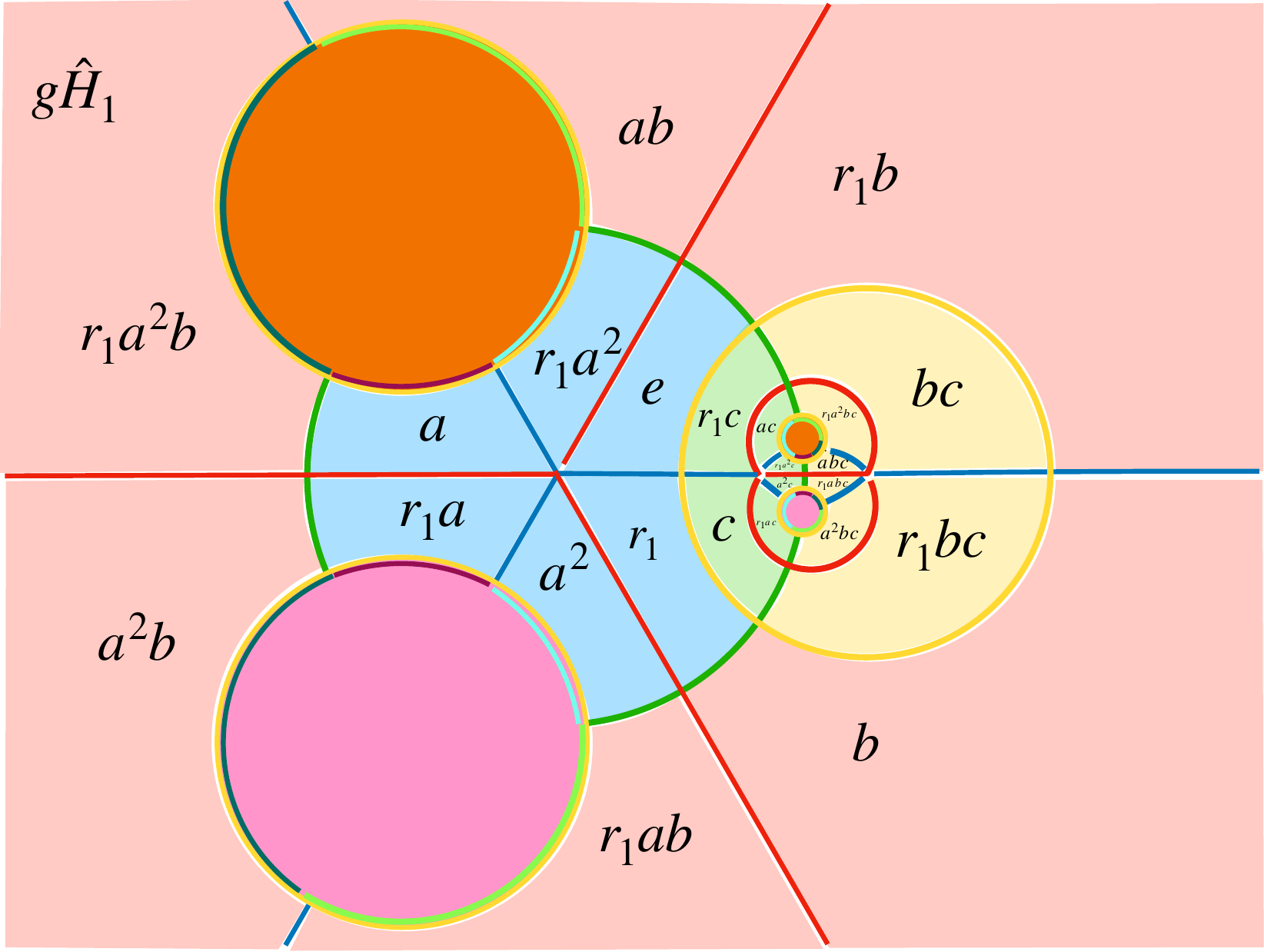}&\qquad\qquad
\includegraphics[width=.4\textwidth,page=41]{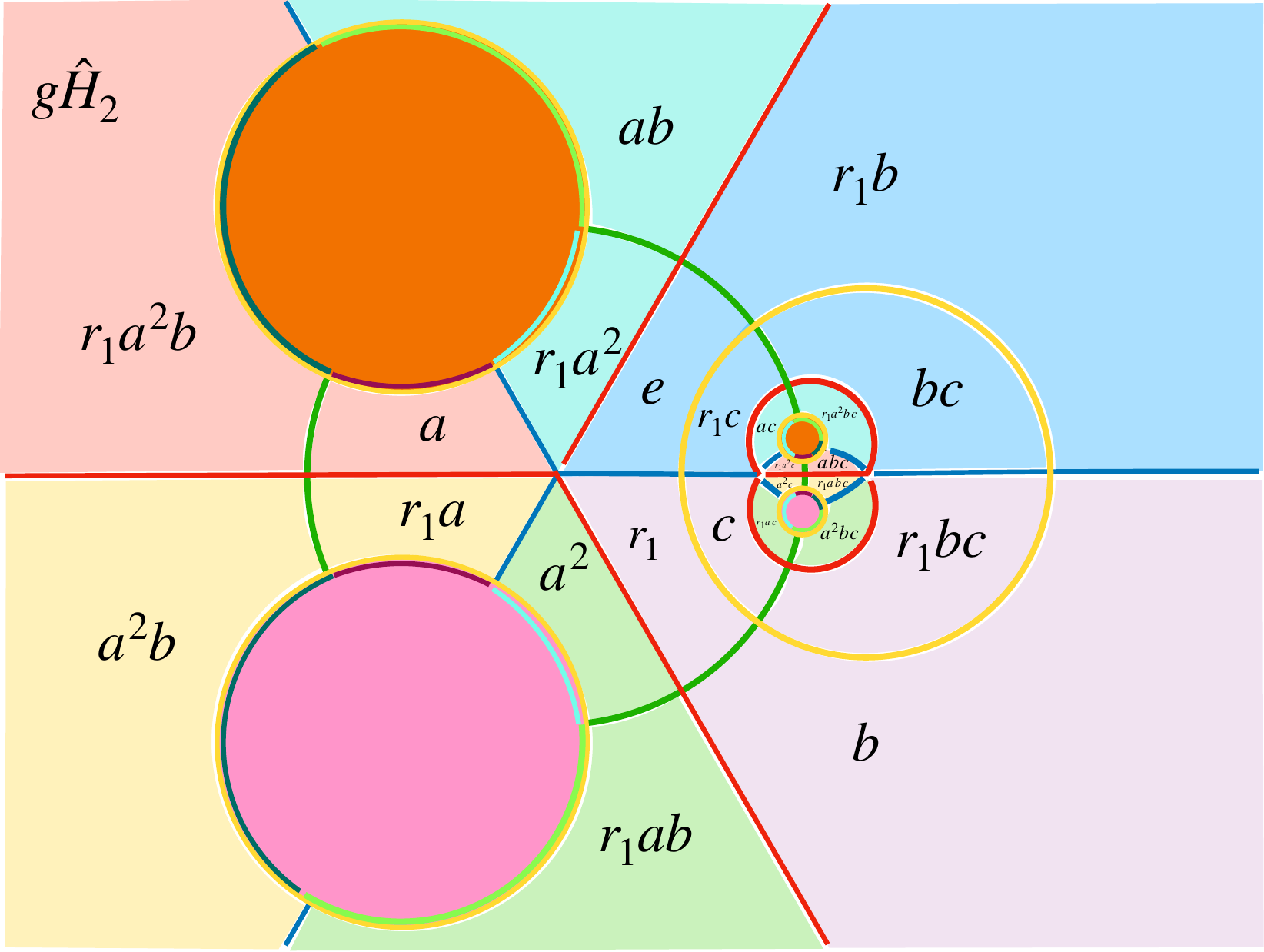}
\end{tabular}
\caption{The decompositions of $\hat H_1$ (L) and $\hat H_2$ (R) into left cosets. The matrix $A$ is determined by looking at the intersection. Here, every colored region in the left panel has a single quadrilateral in common with each colored region in the right panel.}
\label{fig:4partyg2Hcosets} 
\end{figure}

\paragraph{Swapped four-party cluster state} Consider:
\be\label{eq:cluster}
|C_{4s}\rangle=\frac{1}{2}\left(|0000\rangle+|0101\rangle+|1010\rangle-|1111\rangle\right).
\ee
The corresponding subgroups are given by $\hat H_1=\langle r_1,r_1b \rangle$ and $\hat H_2= \langle r_1a^2,r_1c \rangle$, each of which is a $\mathbb{Z}_2\times\mathbb{Z}_2$ subgroup of index 6. We determine their normal cores in $\hat G$ as follows:
\begin{itemize}
    \item $\hat H_1$: From the relations in~(\ref{eq:genus-two-generator-convention})--(\ref{eq:genus-two-extended-group}), $r_1b$ commutes with $r_1$, $a$, $b$, and $c$. (This is because $a(r_1b)a^{-1}=\allowbreak ar_1ba^2=\allowbreak r_1a^2ba^2=\allowbreak r_1b$.) Thus, $r_1b$ is central in $\hat G$ and $\langle r_1b\rangle$ is a normal subgroup of $\hat G$ contained in $\hat H_1$. Notably, any normal subgroup of $\hat G$ contained in $\hat H_1$ must avoid $r_1$ and $b$ because $ar_1a^{-1}=r_1a\notin\hat H_1$ and $aba^{-1}=ba\notin\hat H_1$. Therefore, among the subgroups of $\hat H_1\simeq\mathbb{Z}_2^2$, only $\{e\}$ and $\langle r_1b\rangle$ are normal in $G$. Hence $C_{\hat G}(\hat H_1)=\langle r_1b\rangle\simeq\mathbb{Z}_2$.
    \item $\hat H_2$: Likewise, $r_1c$ is central in $\hat G$, and the conjugates $a(r_1a^2)a^{-1}=r_1\notin\hat H_2$ and $a(ac)a^{-1}=c\notin\hat H_2$ rule out $r_1a^2$ and $ac$ as members of any normal subgroup of $\hat G$ contained in $\hat H_2$, leaving only $\{e\}$ and $\langle r_1c\rangle$. Hence $C_{\hat G}(\hat H_2)=\langle r_1c\rangle\simeq\mathbb{Z}_2$.
\end{itemize}
Taking the quotient~(\ref{eq:Aquotient}) by $\langle r_1b,r_1c\rangle\simeq\mathbb{Z}_2^2$ identifies $b\equiv r_1$ and $c\equiv r_1$. Setting $x\equiv a$ and $y\equiv r_1$, we obtain:
\be
\begin{split}
\tilde G&=\hat G/\langle r_1b,r_1c\rangle=\langle x,y\mid x^3=y^2=(xy)^2=e\rangle\simeq\mathbb{D}_6, \nonumber \\
\tilde H_1&=\hat H_1/C_{\hat G}(\hat H_1)=\{e,y\}\simeq\mathbb{Z}_2, \\
\tilde H_2&=\hat H_2/C_{\hat G}(\hat H_2)=\{e,xy\}\simeq\mathbb{Z}_2. \nonumber
\end{split}
\ee
Note that this triple $(\tilde G,\tilde H_1,\tilde H_2)$ quotient is identical to the $n=3$ instance of~\eqref{eq:core2quotient} in the previous subsection (with $x$ playing the role of $b$ there). Thus, $\tilde A$ is the same $3\times3$ circulant matrix as in~(\ref{circulants}), which yields ${\rm rank}(A)={\rm rank}(\tilde A)=2$. Substituting $|\hat G/\hat H_1|=|\hat G/\hat H_2|=6$ and ${\rm rank}(A)=2$ into~(\ref{finalcoarse}), we obtain:
\be
\CE(e,a^{-1},b,c;\ket{C_{4s}})=2^{-24+6+6-2}=2^{-14}.
\ee
This contrast between the computations for states (\ref{eq:C4}) and (\ref{eq:cluster}) demonstrates that the genus-2 multi-invariant is sensitive to the permutations on subsystems. 

\section{Discussion}
\label{sec:discussion}
Motivated by recent works on multi-invariants that originated in the AdS/CFT correspondence, here we begin an exploration of explicit calculations of these quantities. We study stabilizer states because of their tractability and diverse applications. Further borrowing from the holographic context, we concentrate on those patterns of index contractions, which can be understood with reference to regular tessellations of two-dimensional manifolds. 

With this restriction on the ``jigsaw pattern,'' the calculation of multi-invariants of simple stabilizer states becomes amenable to a group theoretic approach. By ``simple states,'' we mean states with two GHZ-like cores and some close cousins thereof such as, for example, state~(\ref{eq:cluster}). A key technical tool is to organize the information about the phases of non-vanishing terms in the index contraction by double cosets of certain subgroups of the extended replica group. 

\paragraph{Scope of method} The key thing, which makes our computational method work, is that the phases of non-vanishing combinations of terms enter the multi-invariant depend on a single bilinear. This is explicit, for example, in equation~(\ref{eq:arf}). That fact, in turn, descends from the construction of core-2 states by the controlled-$Z$ gate $(-1)^{(1-Z) \otimes (1-Z)/4}$. This bilinear structure occurs frequently in physics literature and is generally associated with the group cohomology of $\mathbb{Z}_2 \times \mathbb{Z}_2$ \cite{TP1, GT1} and, by extension, to symmetry-protected topological order \cite{BS1}. The two $\mathbb{Z}_2$ factors are manifest in equation~(\ref{color2bits}). It would be interesting to see if multi-invariants of states with symmetry protected topological (SPT) order might more generally be evaluated by similar methods. We leave this speculation to future work. 

A more pedestrian extension of the scope of this work is to consider states whose phases are still quadratic or bilinear but involve more terms. We consider this in another paper, which is submitted to the arXiv concurrently \cite{SMI2}. There, we concentrate on multi-R\'enyi entropies whose replica group is $\mathbb{Z}_N^{\tq -1}$. For that replica group, the algebraic approach developed in the present paper is an overkill. On the other hand, an alternative methodology based on Gauss sums over finite fields allows us to carry out the computation for a significantly broader class of stabilizer states. Our two papers can be viewed as complementary: the present one explores a broad class of replica groups on a narrow set of states while \cite{SMI2} does the opposite.

\subsection{More than one qubit per party}
\label{sec:moreparty}
Thus far, we have concentrated on $\tq$-partite multi-invariants $\tq$-qubit states. That is, every party $A_a$ consisted of an individual qubit. We now outline what happens when you collect more than one qubit in a subsystem. 

Because multi-invariants are---indeed---{\it invariant} under local unitary transformations, we may act with arbitrary gates with each party without affecting the invariant. In particular, we can apply $CZ = (-1)^{(1-Z) \otimes (1-Z)/4}$ gates to entangle or disentangle the qubits within a subsystem. This operation often allows us to reduce computations of more complicated stabilizer states into smaller components. We may formalize this with the following definition:

\begin{definition}
    Consider a stabilizer state whose Clifford-equivalent graph state representation is connected.\footnote{Given a graph with $N$ vertices and a bijection between the vertices and $N$ qubits, the corresponding {\bf graph state} is prepared from $|+\rangle^{\otimes N}$ by applying $CZ$ gates on those pairs of qubits, which are connected by edges in the graph. Every stabilizer state is related to at least one graph state by a qubit-wise Clifford transformation~\cite{PhysRevA.69.022316}.} Given a partition of the constituent qubits into parties (subsystems) $A_a$, we say that the state is \textbf{reducible} with respect to that partition if unitary transformations acting subsystem-wise can render the underlying graph disconnected.
\end{definition}

\noindent
The simplest example is the GHZ state on $\tq$ parties, which is reducible with respect to every partition of its qubits except the $\tq$-partite one-qubit-per-subsystem partition. This is illustrated in Figure~\ref{fig:GHZpackaged}:

\begin{figure}[H]
  \centering
   \includegraphics[width=.9\textwidth]{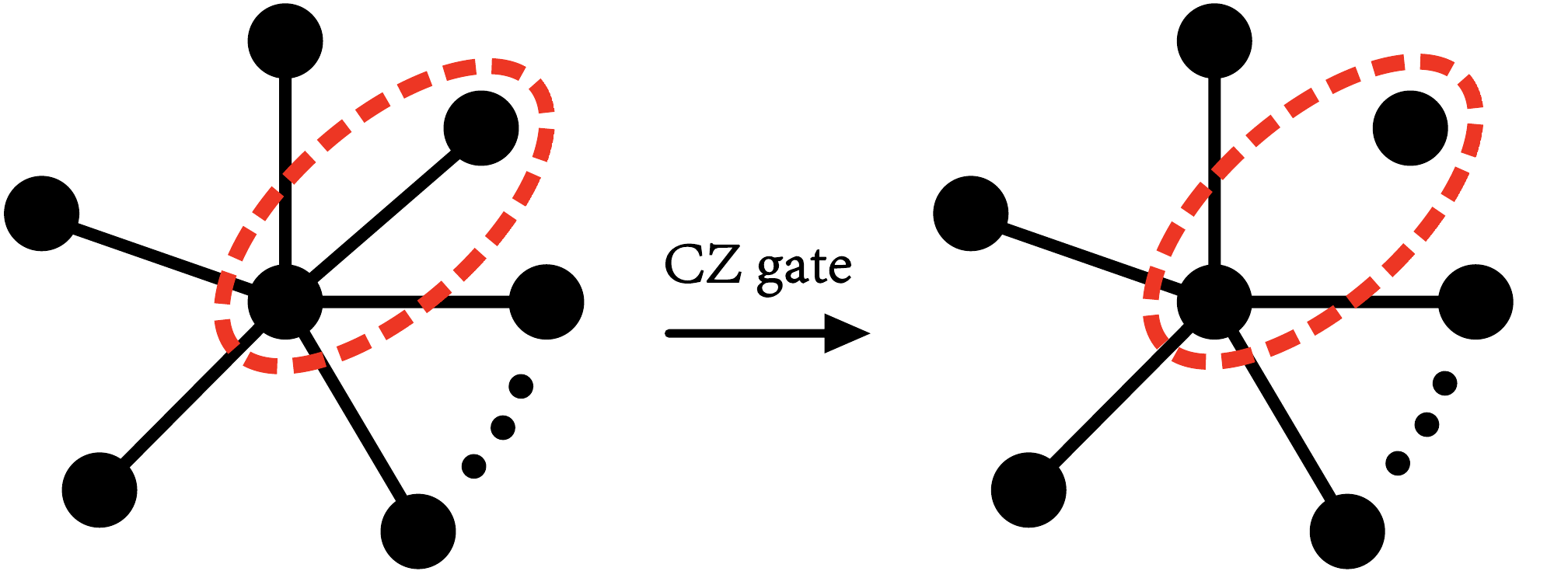}
  \caption{A GHZ state is locally Clifford equivalent to a star graph state. Due to the symmetry of $|{\rm GHZ}_{\tq}\rangle$, every qubit can play the role of the central node. Applying a CZ gate on the two highlighted qubits maps the state to $|{\rm GHZ}_{\tq-1}\rangle \otimes |+\rangle$.}
\label{fig:GHZpackaged}
\end{figure}
\noindent
A slightly less trivial class of examples comprises core-2 states~(\ref{eq:defn_core2state}). One graph representation of state $\ket{\text{core-}2_{p+q}}$ looks like two star graphs with $p-1$ and $q-1$ arms, respectively, whose central nodes are connected by an additional edge. To see that it is reducible with respect to every partition except the finest one ($1^{p+q}$), we consider a pair of qubits in one of two cases: 
\begin{itemize}
    \item Both qubits belong to the same block of either $p$ or $q$ qubits. This case reduces to that shown in Figure~\ref{fig:GHZpackaged}.
    \item One qubit belongs to the block of $p$ qubits and the other to the block of $q$ qubits. By the symmetry of $\ket{\text{core-}2_{p+q}}$, those qubits can serve as the central nodes of the two stars in the graph representation. Then applying a $CZ$ on them disconnects the two stars.
\end{itemize}

\paragraph{Reducing the calculation to smaller invariants}
So long as the graph representation of a stabilizer state is reducible, we can reduce the calculation of its multi-invariants to smaller, component calculations. This is due to two properties of multi-invariants:
\begin{itemize}
    \item Any multi-invariant that does not involve the maximally fine partition of qubits ($1^{\tq}$) can be rewritten as some other multi-invariant with the maximally fine partition. For example:
    \begin{equation}
S_{n}^{(p+q-1)}(A_1:...:A_{p-1}:A_pB_1:B_2:...:B_q)
=S_{n}^{(p+q)}(A_1:...:A_p:B_1:...:B_q)
    \end{equation}
    \item A multi-invariant of a $p+q$ partite product state $\rho_{AB}=\rho_{A}\otimes\rho_{B}=\rho_{A_1...A_s}\otimes\rho_{B_1...B_t}$ always decomposes into a sum of some multi-invariants of $\rho_{A_1...A_s}$ and $\rho_{B_1...B_t}$, evaluated separately. For example:
    \begin{equation}
        S_{n}^{(p+q)}(A_1:...:A_s:B_1:...:B_t)|_{\rho_A \otimes \rho_B}
        =S_{n}^{(s)}(A_1:...:A_s)|_{\rho_{A}}+S_{n}^{(t)}(B_1:...:B_t)|_{\rho_{B}}
        \label{eq:red}
    \end{equation}
    On the right hand side, we set the `one-partite multi-entropy' $S_{n}^{(1)}=0$ by convention. We emphasize that the multi-invariants after the reduction may not have a direct relationship to the initial multi-invariant, like they do in example~(\ref{eq:red}). 
\end{itemize}
So long as the state is reducible, step~(\ref{eq:red}) can be taken and the computation splits up into smaller components.

\paragraph{Example}
For illustration, consider $S_{n}^{(4)}(A_{1}A_2:A_3B_1:B_2:B_3)$ of $|\text{core-}2_{3+3}\rangle$. Applying the above steps sequentially yields:
\begin{equation}
    \begin{aligned}
S_{n}^{(4)}(A_{1}A_2:A_3B_1:B_2:B_3)|_{\text{core-}2_{3+3}}&=S_{n}^{(4)}(A_{1}:A_3B_1:B_2:B_3)|_{\text{core-}2_{2+3}}\\
&=S_{n}^{(4)}(A_{1}:A_3B_1:B_2:B_3)|_{\text{Bell}\otimes\text{GHZ}_3}\\
&=S_{n}^{(2)}(A_{1}:A_3)|_{\text{Bell}}+S_{n}^{(3)}(B_1:B_2:B_3)|_{\text{GHZ}_3}\\
&=2+\frac1n.
    \end{aligned}
\end{equation}
For stabilizer states beyond core-$2$ states, we can prove that there always exists at least one partition of qubits with respect to which the state is not reducible. In such cases, the calculations of multi-invariants become harder but remain doable. We discuss this in forthcoming work \cite{SMI2}.

\section*{Acknowledgments}
We are grateful to Abhijit Gadde and Sriram Akella for useful discussions and for synchronizing the arXiv submissions of our papers. XXJ thanks Ya-Wen Sun for useful discussions. BC thanks the organizers of the workshops ``Discussions on Quantum Spacetime'' held at LMSI (Mumbai, India)
and ``Quantum Gravity, Chaos, and Observers for Black Holes and Cosmology'' held at the Aspen Center for Physics, where some of this work was completed. The work of BC, YF, and XXJ is supported by an NSFC grant number 12042505. The work of JH is supported by MEXT KAKENHI Grant-in-Aid for Transformative Research Areas (A) through the ``Extreme Universe'' collaboration: Grant Number 21H05187. The work of XXJ is also supported by the Shuimu Tsinghua Scholar Program of Tsinghua University.

\appendix

\section{Properties of the left coset intersection matrix}
\label{sec:leftcosetprop}

In this appendix we deduce some of properties and structure of the left coset intersection matrix $A$. Throughout we will assume that the two subgroups $\hat H_1$ and $\hat H_2$ have the additional property that $\hat H_1\cap\hat H_2=\{e\}$. In particular this places some strict constraints on $A$ as a consequence of the following theorem:
\begin{theorem}
    Let G be a group and $H,K$ subgroups of $G$ then $H\cap K=\{e\}$ iff $\forall g_1,g_2 \in G$ $g_1H\cap g_2K$ is either empty or a singleton.
\end{theorem}
\begin{proof}
    Suppose $g_1H\cap g_2K$ contains at least two elements let these be $x,y$. This means that there exists $x,y\in G$ and $h_1,h_2\in H$ and $k_1,k_2\in K$ such that $x=g_1h_1=g_2k_1$ and $y=g_1h_2=g_2k_2$ but rearranging then we have
    \be
    \begin{split}
     h_1=g_1^{-1}g_2k_1\\
        h_2=g_1^{-1}g_2k_2
    \end{split}
    \ee
and considering the product we find
\be
h_1^{-1}h_2=k_1^{-1}g_2^{-1}g_1g_1^{-1}g_2k_2=k_1^{-1}k_2.
\ee
Now $h_1^{-1}h_2\in H$ and $k_1^{-1}k_2\in K$ thus $h_1^{-1}h_2=k_1^{-1}k_2\in H\cap K$. By assumption $H\cap K=\{e\}$ which means $h_1=h_2$ and $k_1=k_2$. In other words $x=y$ and $g_1H\cap g_2K$ contains at most one element, a contradiction.
\end{proof}
This is actually a special case of the following more general theorem:
\begin{theorem}
Let G be a group and $H,K$ subgroups of $G$ then $\forall g_1,g_2\in G$ $g_1H\cap g_2K$ is either empty or a left coset of $H\cap K$.
\end{theorem}
So in particular since $\hat H_1\cap \hat H_2=\{e\}$ then each element of $G$ is its own unique coset of $H\cap K$ and the intersection of the left cosets is at most one element.\footnote{One can define the matrix $A$ for any choice of subgroups of any finite group (independent of applications to multi-invariants). If we look at more general subgroups with non-trivial overlap then the matrix $A$ will consists only of the entries $0$ and $|\hat H_1 \cap \hat H_2|$ $\mod2$ as a result of theorem 2. So if the cardinality of the intersection of the subgroups is even then every entry in the intersection matrix is 0.}

Now we can list the following properties of $A$:
\begin{itemize}
\item The number of rows and columns is $m_1$ and $m_2$.
    \item Since cosets partition the group given a group element $g\in\hat G$ it will belong to exactly one coset of $\hat H_1$ and one of $\hat H_2$. Thus each group element will contribute exactly one non-zero entry to $A$ and we can conclude that in total $A$ contains exactly $2s$ non-zero entries.
    
    \item The columns of $A$ contains exactly $2l_1$ ones and each row contains exactly $2l_2$ ones. Note that since $2s=2l_1m_1=2l_2m_2$ we find that the total number of ones is $2s$ as expected. Thus, we have a binary matrix with fixed row and column sums.

    \item Let $c_i\hat H_2\cap d_j\hat H_1\neq\{\emptyset\}$ then there exists $k\in \hat H_2$ and $h\in \hat H_1$ such that $c_ik=d_jh$ equivalently rearranging $c_i^{-1}d_j=kh^{-1}$. Now by definition $kh^{-1}$ is an element of the double coset $\hat H_2e\hat H_1$ meaning that so is $c_i^{-1}d_j$.
    \begin{itemize}
        \item Thus: $a_{ij}=1$ iff $c_i^{-1}d_j\in \hat H_2e\hat H_1$\\
    \end{itemize}

    \item Suppose we know an entire column of the matrix $A$ wlog let it be $a_{i1}$. Then this is enough information to construct the entire matrix. Consider $c_i\hat H_2\cap \hat H_1$ and left multiple $\hat H_1$ by an element of $\{d_j\}$ this moves us from one column to another. In particular we now need to look at the intersection $c_i\hat H_2\cap d_j\hat H_1$. If we multiple we then have $d_j^{-1}c_i\hat H_2\cap \hat H_1$. So either $c_i\hat H_2\cap d_j\hat H_1$ and $d_j^{-1}c_i\hat H_2\cap \hat H_1$ are either both empty or non-empty. In particular we have $a_{ij}=a_{k1}$ where $c_k=d_j^{-1}c_i$.

    Thus, the left action of $g$ on $H$ has the effect of permuting the original column to generate the rest of the matrix. 
    \end{itemize}

    Since our ultimate goal is to determine the rank of $A$ often times the matrix $A$ will contain redundant information in the form of repeated rows and columns which are not linearly independent it is therefore advantageous to understand under what conditions cosets will generate the same row/column.

    Important for us is the following result: Take an element $g\in \hat{G}$ and  coset $d_i\hat H_1$ and consider the left multiplication. For which elements of $g$ will the coset be unchanged? That is:
\be
gd_j\hat H_1=d_j\hat H_1.
\ee
These elements are the stabilizers of the coset which are the conjugates $d_i\hat H_1d_i^{-1}$. Now if we want those group elements which are stabilizers of \emph{every} coset then this defines a subgroup called the \emph{core} of $\hat{H_1}$ in $\hat G$: $C_{\hat G}(\hat H_1)$.

    \begin{definition}[Core of a subgroup] For a group $G$ the core of a subgroup $H$ is the largest subgroup of $H$ which is normal in $G$. Equivalently
    \be
    C_{G}(H)=\cap_{g\in G}gHg^{-1}
    \ee
    \label{defcore}
    \end{definition}

So now let $x\in C_{\hat{G}}(\hat H_2)$ and consider for fixed $d_j$ the column vector corresponding to $c_i\hat{H}_2\cap d_j\hat H_1$ and left multiple $d_j\hat H_1$ by $x$. Then equivalently we have
\be
 x^{-1}c_i\hat H_2\cap d_j\hat H_1=c_i\hat H_2\cap d_j\hat H_1 \quad \forall i
\ee
so the two vectors are the same component by component. Since the vectors are the same they are not linearly independent. The same is also true for the row vectors by considering the core $C_{\hat G}(\hat H_1)$

Thus, we can define the new subgroups $\langle\hat H_1,C_{\hat{G}}(\hat H_2)\rangle$ and $\langle\hat H_2,C_{\hat{G}}(\hat H_1)\rangle$. What this has the effect of doing is combining all of the cosets of the original subgroups which had the same column/row vector. So if we examine the new intersection matrix it will have the same rank, but all row vectors will be unique (no repeats) and all column vectors will be unique (no repeats).

Now we are in the situation that $\langle C_{\hat{G}}(\hat H_1),C_{\hat{G}}(\hat H_2)\rangle$ is a normal subgroup of $\langle\hat H_1,C_{\hat{G}}(\hat H_2)\rangle$, $\langle \hat H_2,C_{\hat{G}}(\hat H_1)\rangle$ and $\hat G$ so we can take a quotient and instead consider the groups
\be\label{eq:Aquotient}
\begin{split}
\tilde G&=\hat G/\langle C_{\hat{G}}(\hat H_1),C_{\hat{G}}(\hat H_2)\rangle\\
\tilde H_1&=\langle \hat H_1,C_{\hat{G}}(\hat H_2)\rangle/\langle C_{\hat{G}}(\hat H_1),C_{\hat{G}}(\hat H_2)\rangle=\hat H_1/C_{\hat{G}}(\hat H_1)\\
\tilde H_2&=\langle \hat H_2,C_{\hat{G}}(\hat H_1)\rangle/\langle C_{\hat{G}}(\hat H_1),C_{\hat{G}}(\hat H_2)\rangle=\hat H_2/C_{\hat{G}}(\hat H_2)
\end{split}
\ee
with the final result
\begin{theorem}\label{thm:atilde} Let the groups $\hat G,\hat H_1,\hat H_2$ and $\tilde G,\tilde H_1,\tilde H_2$ be defined as above with left coset intersection matrices $A,\tilde A$
\begin{itemize}
    \item The rank of the two matrices is the same:  \be
rank(A)=rank(\tilde A)
\ee
\item The row vectors of $\tilde A$ are all unique. The column vectors of $\tilde A$ are all unique.
\end{itemize}

\end{theorem}

\paragraph{The case of normal subgroups} When the subgroups $\hat H_1$ and $\hat H_2$ are normal subgroup of $\hat G$ we can derive a particularly simple answer for $rank(A)$.

As a reminder here is the definition for normal subgroup:
\begin{theorem}
Let $G$ be a group and $N<G$ a subgroup. Let the sets $X$ and $Y$ generate the groups that is $N=\langle X\rangle$ and $G=\langle Y\rangle$. Then the following are equivalent:
\begin{itemize}
    \item $N$ is a normal subgroup of $G$ written $N\lhd G$
    \item $\forall\; g \in G$ and $n\in N: \quad gng^{-1}\in N$
    \item $\forall y\in Y$ and $x\in X:$ $yxy^{-1}$ and $y^{-1}xy\in N$
\end{itemize}
\end{theorem}
Normal subgroups are important because they allow us to define quotient groups. For normal subgroups there is no distinction between left and right cosets because they are the same. The last line allows one to determine if a subgroup is normal using a finite number of multiplications by checking only the generators.

First let us define the double cosets:
\begin{definition}
Let $G$ be a group and $H$ and $K$ subgroups. A set of the form
\be
HgK=\{hgk:h\in H,K\in K\}
\ee
is a double coset.
\end{definition}
A few properties of double cosets:
\begin{itemize}
\item Double cosets are disjoint and cover the entire group $G$. 
\item Each double coset is the union of some subset of the left cosets of $K$. Respectively the union of some subset of the right cosets of $H$.
\item $g_1K$ and $Hg_2$ have non-empty intersection iff they lie in the same double coset $Hg_2g_1K$.
\item Unlike cosets the number of double cosets need not divide the order of the group.
\end{itemize}
In \cite{2013arXiv1304.6111B} the left-right coset intersection matrix $B$ of $gH$ and $Kg$ was considered. It is shown that the resulting matrix always consists of blocks of all ones where the number of blocks is precisely the number of double cosets. Each block gives rise to a single linear independent vector and it can be concluded that ${\rm rank}(B)=|K\setminus G\,/\,H|$.

So now:
\begin{itemize}
\item If one of the subgroups, say $K$, is normal then the left cosets $gK$ and the right cosets $Kg$ are the same. As a consequence the two matrices $A,B$ are the same and we have ${\rm rank}(A)=|K\setminus G\,/\,H|$.  The final answer for the multi-invariant \eqref{eq:4clusterfinal} becomes
\be
\begin{split}
=-\left(|\hat G|-\frac{|\hat G|}{|\hat H_1|}-\frac{|\hat G|}{|\hat H_2|}+|\hat H_2\setminus\hat G\,/\,\hat H_1|\right)\log(2)
\end{split}.
\ee

\item 
If both of $H,K$ are normal subgroups of $G$ then the identity double coset $KeH$, that is the set $\{kh|h\in H,k \in K\}$, is a normal subgroup of $G$ and the smallest subgroup of $G$ containing both $H$ and $K$. This induces an isomorphism between the double cosets $K\setminus G/H$ and the quotient group $ G/KH$. In particular $|KH|=\frac{|H||K|}{|H\cap K|}=|H||K|$ and we have $|K\setminus G\,/\,H|=[G:KH]=\frac{|G|}{|H||K|}$. the logarithm of~\eqref{finalcoarse} further simplifies to
\be\label{eq:bothnormal}
\begin{split}
=-|\hat G|\left(1-\frac{1}{|\hat H_1|}-\frac{1}{|\hat H_2|}+\frac{1}{|\hat H_1||\hat H_2|}\right)\log(2)
\end{split}.
\ee
\end{itemize}

One might be tempted to propose that even in the case of non-normal subgroups it would be generically true that for any finite group and subgroups $rank(A)=|K\setminus G/H|$ however this is false as shown by the following simple counter example: We take the dihedral group of order eight

\be
\mathbb{D}_8=\langle a,b|a^4,b^2,(ba)^2\rangle
\ee
and the two non-normal subgroups
\be
H=\{e,b\}, \quad K=\{e,ab\}
\ee
then
\be
\begin{split}
    g H:\quad e& H =\{e,b\}\\
    a& H =\{a,ab\}\\
    a^2& H =\{a^2,a^2b\}\\
    a^3& H =\{a^3,a^3b\}\\    
\end{split}
\ee
\be
\begin{split}
    g K:\quad e& H =\{e,ab\}\\
    a& K =\{a,a^2b\}\\
    a^2& K =\{a^2,a^3b\}\\
    a^3& K =\{a^3,b\}\\    
\end{split}
\ee
There are two double cosets
\be
\begin{split}
    & Ke H=\{e,,b,ab,a\}\\
    & Ka^2H=\{a^2,a^2b,a^3b,a^3\},\\
\end{split}
\ee
but the left coset intersection matrix is:
\be
A=\begin{pmatrix}
1 & 1 & 0 & 0 \\
0 & 1 & 1 & 0 \\
0 & 0 & 1 & 1  \\
1 & 0 & 0 & 1  
\end{pmatrix}
\ee
We used $e,a,a^2,a^3$ as representatives of cosets, which corresponds to the rows and columns. This matrix has rank 3 over $\mathbb{F}_2$.

\section{Rank of circulant matrices}
\label{app:circ}
Matrices~(\ref{circulants}) are called \emph{circulant matrices}. Their defining property is that all rows are given by sequential cyclic permutations of one another. Here we compute their ranks.

Writing the first row vector in the form $(c_0,c_{n-1},\cdots,c_1)$, construct the polynomial:
\be
f(x)=c_0+c_1x+\cdots c_{n-1}x^{n-1}
\ee
As shown in Reference~\cite{cirmatrank}, the rank is then given by:
\be
{\rm rank}(A)=n-\deg(\gcd(f(x),x^n-1))
\label{rankfromref}
\ee
We remind the reader that we are working over $\mathbb{F}_2$, so for example $2x^2+x-1=x+1$. Thus, the greater common divisor in equation~(\ref{rankfromref}) has to be handled with some care.

We use the following field-independent facts about the greatest common divisor:
\begin{itemize}
    \item For any polynomials $p,q,r$, $\gcd(p,q+rp)=\gcd(p,q)$.
    \item Given a polynomial $p(x)$, $(x-a)$ divides $p(x)$ if and only if $p(a)=0$.
\end{itemize}
Applied to matrices~(\ref{circulants}), we have $f(x)=x^{n-1}+1$ and $g(x)=x^n-1$. Considering the combination $x f(x) + g(x) = 2x^n + x - 1 = x-1$, the rank is $n - \deg(\gcd(x^{n-1}+1, x-1)$, which can only be $n$ or $n-1$. Using the second property, we find that $x-1$ divides $f(x)$ because $f(1)=0$. Therefore:
\be
{\rm rank}(A)= n - 1.
\ee

\bibliography{main.bib}
\end{document}